\documentclass[12pt,draftclsnofoot, perreview, onecolumn]{IEEEtran}%Prof. Ma's template

\usepackage{bbm}
\usepackage{color}
\usepackage{amsmath}
\usepackage{mathrsfs}
\usepackage{pifont}
\usepackage{amsfonts}
\usepackage{citesort}
\usepackage{mathdots}
\usepackage{graphicx}
\usepackage{threeparttable}

\ifCLASSINFOpdf
\else
\fi

\ifCLASSOPTIONcompsoc
  \usepackage[tight,normalsize,sf,SF]{subfigure}
\else
  \usepackage[tight,footnotesize]{subfigure}
\fi

\newtheorem{theorem}{Theorem}
\newtheorem{definition}{Definition}
\newtheorem{algorithm}{Algorithm}

\newtheorem{lemma}{Lemma}

\begin{document}
\title{An Information-Spectrum Approach to the Capacity Region of the Interference Channel}

\author{Authors}

\author{Xiao~Ma,~\IEEEmembership{Member,~IEEE,}
        Lei~Lin,
        Chulong~Liang,
        Xiujie~Huang,~\IEEEmembership{Member,~IEEE,}
        and~Baoming~Bai,~\IEEEmembership{Member,~IEEE}
\thanks{This paper was presented in part at the 2012 IEEE
International Symposium on Information Theory. This work was supported by the 973 Program~(No.2012CB316100) and the NSFC~(No.61172082).}
\thanks{X.~Ma, and C.~Liang are with the Department of Electronic and Communication Engineering, Sun Yat-sen University, Guangzhou 510006,
Guangdong, China~(email: maxiao@mail.sysu.edu.cn).}
\thanks{L.~Lin is with the Department of Applied Mathematics, Sun Yat-sen University, Guangzhou 510006,
Guangdong, China~(email: linlei2@mail2.sysu.edu.cn).}
\thanks{X.~Huang is with the Department of Electrical Engineering, University of Hawaii, Honolulu 96822, HI, USA.}
\thanks{B.~Bai is with the State Lab. of ISN, Xidian University, Xi'an 710071, Shaanxi, China.}}

%\markboth{SUBMITTED TO IEEE TRANSACTIONS ON INFORMATION THEORY,
%MARCH, 2009}{Huang \MakeLowercase{\textit{et al.}}: Upper Bounds on the Capacities of Non-Controllable Finite-State
%Channels with/without Feedback}




\maketitle\thispagestyle{empty}

\begin{abstract}
In this paper, we present a general formula for the capacity region of a general interference channel with two pairs of users. The formula shows that the capacity region is the union of a family of rectangles, where each rectangle is determined by a pair of spectral inf-mutual information rates. Although the presented formula is usually difficult to compute, it provides us useful insights into the interference channels. In particular, when the inputs are discrete ergodic Markov processes and the channel is stationary memoryless, the formula can be evaluated by BCJR algorithm. Also the formula suggests us that the simplest inner bounds~(obtained by treating the interference as noise) could be improved by taking into account the structure of the interference processes. This is verified numerically by computing the mutual information rates for Gaussian interference channels with embedded convolutional codes. Moreover, we present a coding scheme to approach the theoretical achievable rate pairs. Numerical results show that decoding gain can be achieved by considering the structure of the interference.

%This paper is concerned with general interference channels characterized by a sequence of transition~(conditional) probabilities.
%We present a general formula for the capacity region of the interference channel with two pairs of users. The formula shows that the capacity region is the union of a family of rectangles, where each rectangle is determined by a pair of spectral inf-mutual information rates. Although the presented formula is usually difficult to compute, it provides us useful insights into the interference channels. For example, the formula suggests us that the simplest inner bounds~(obtained by treating the interference as noise) could be improved by taking into account the structure of the interference processes. This is verified numerically by computing the mutual information rates for Gaussian interference channels with embedded convolutional codes.
\end{abstract}


% \fi
\IEEEpeerreviewmaketitle

\section{Introduction}
The interference channel~(IC) is a communication model with multiple pairs of senders and receivers, in which each sender has an independent message intended only for the corresponding receiver. This model was first mentioned by Shannon~\cite{Shannon61} in 1961 and further studied by Ahlswede~\cite{Ahlswede74} in 1974. A basic problem for the IC is to determine the capacity region, which is currently one of long-standing open problems in information theory. Only in some special cases, the capacity regions are known, such as strong interference channels, very strong interference channels and deterministic interference channels~\cite{Carleial75,Han81,ElGamal82,Costa87}. For a general IC, various inner and outer bounds of the capacity region have been obtained. In 2004, Kramer derived two outer bounds on the capacity region of the general Gaussian interference channel~(GIFC)~\cite{Kramer04}. The first bound for a general GIFC unifies and improves the outer bounds of Sato~\cite{Sato77} and Carleial~\cite{Carleial83}. The second bound follows directly from the outer bounds of Sato~\cite{Sato78} and Costa~\cite{Costa85}, which is derived by considering a degraded GIFC and is even better than the first one for certain weak GIFCs. The best inner bound~(the so-called HK region)
 is that proposed by Han and Kobayashi~\cite{Han81}, which has been simplified by Chong {\em et~al.} and Kramer in their independent works~\cite{Chong08} and~\cite{Kramer06}. In recent years, Etkin, Tse and Wang~\cite{Etkin08} showed by introducing the idea of approximation that  HK region~\cite{Han81} is within one bit of the capacity region for the GIFC.

In~\cite{Huang2011}, the authors proposed a new computational model for the two-user GIFC, in which one pair of users~(called {\em primary users}) are constrained to use a fixed encoder and the other pair of users~(called {\em secondary users}) are allowed to optimize their code. The maximum rate at which the secondary users can communicate reliably without degrading the performance of the primary users is called the {\em accessible capacity} of the secondary users. Since the structure of the interference from the primary link has been taken into account in the computation, the accessible capacity is usually higher than the maximum rate when treating the interference as noise, as is consistent with the spirit of~\cite{FrancoisArticle2011}\cite{Moshksar11}. However, to compute the accessible capacity~\cite{Huang2011}, the primary link is allowed to have a non-neglected error probability. This makes the model unattractive when the capacity region is considered. For this reason, we will relax the fixed-code constraints on the primary users in this paper. In other words, we will compute a pair of transmission rates at which both links can be asymptotically error-free.

In this paper, we consider a more general interference channel which is characterized by a sequence of transition probabilities.
By the use of the information spectrum approach~\cite{Han93}\cite{Han03},  we present a general formula for the capacity region of the general
interference channel with two pairs of users. The formula shows that the capacity region is the union of a family of rectangles, in which each rectangle
is determined by a pair of spectral inf-mutual information rates. The information spectrum approach, which is based on the
\emph{limit superior/inferior in probability} of a sequence of random variables, has been proved to be powerful in characterizing the
limit behavior of a general source/channel.
%The spectral inf-mutual information rate is defined based on the {\em limit superior/inferior in probability}, which has been proved to be a powerful tool to characterize the limit behavior of a general process.
For instance, in~\cite{Han93} and~\cite{Verdu94}, Han and Verd\'{u} proved
that the minimum compression rate for a general source equals its \emph{spectral sup-entropy rate} and
the maximum transmission rate for a general point-to-point channel equals its \emph{spectral inf-mutual information rate} with an optimized input process. Also the information spectrum approach can be used to derive the capacity region
of a general multiple access channel~\cite{Han98}. For more applications of the information spectrum approach, see~\cite{Han03} and the references therein.

The rest of the paper is structured as follows. Sec.~\ref{Definition_section} introduces the definition of a general IC and the concept of the spectral inf-mutual information rate. In Sec.~\ref{Capacity_section}, a general formula for the capacity region of the general IC is proposed; while, in Sec.~\ref{Sec:algFortheorem}, a trellis-based algorithm is presented to compute the
pair of rates for a stationary memoryless IC with discrete ergodic Markov sources.
In Sec.~\ref{GIFC_section}, numerical results are presented for a GIFC with binary-phase shift-keying~(BPSK) modulation. Sec.~\ref{Sec:decAnddetAlg} provides the detection and decoding algorithms for channels with structured interference. Sec.~\ref{Conclusions} concludes this paper.

In this paper, a random variable is denoted by an upper-case letter, say $X$, while its realization
and sample space are denoted by $x$ and $\mathcal{X}$, respectively. The sequence of random variables with length $n$ are
denoted by $X^n$, while its realization is denoted by ${\bf x} \in \mathcal{X}^n$ or $x^n \in \mathcal{X}^n$. We use $P_X(x)$ to denote the probability mass function~(pmf) of $X$ if it is discrete or the probability density function~(pdf) of $X$ if it is continuous.

%-------------------------------------------------------------------------------------------
\begin{figure}
  \centering
  \includegraphics[width=11.0cm]{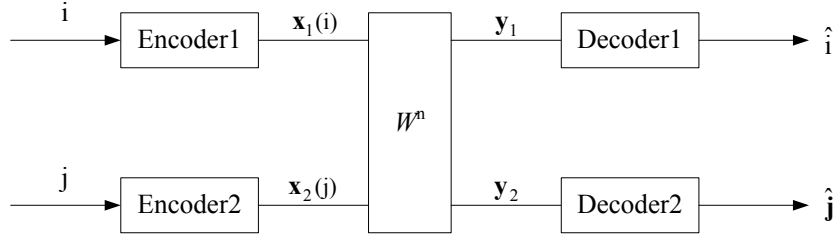}\\
  \caption{General interference channel ${\bf W}$.}\label{general IFC_Fig}
\end{figure}

%-------------------------------------------------------------------------------------------
\section{Basic Definitions And Problem Statement}\label{Definition_section}
%In this section we shall define the general interference channel and present the preliminaries.
\subsection{General IC}
Let $\mathcal{X}_1$, $\mathcal{X}_2$ be two finite input alphabets and $\mathcal{Y}_1$, $\mathcal{Y}_2$ be two finite output alphabets. A general interference channel ${\bf W}$~(see Fig.~\ref{general IFC_Fig}) is characterized by a sequence ${\bf W} = \{W^n(\cdot,\cdot|\cdot,\cdot)\}_{n=1}^{\infty}$, where $W^n: \mathcal{X}_1^n \times \mathcal{X}_2^n \rightarrow \mathcal{Y}_1^n \times \mathcal{Y}_2^n$ is a probability transition matrix. That is, for all $n$,
\begin{eqnarray*}
  W^n({\bf y}_1,{\bf y}_2|{\bf x}_1,{\bf x}_2)  &\geq&  0 \\
  \sum\limits_{{\bf y}_1 \in \mathcal{Y}_1^n,{\bf y}_2 \in \mathcal{Y}_2^n} W^n({\bf y}_1,{\bf y}_2|{\bf x}_1,{\bf x}_2) &=& 1.
\end{eqnarray*}
The marginal distributions $W_1^n,W_2^n$ of the $W^n$ are given by
\begin{eqnarray}
W_1^n({\bf y}_1|{\bf x}_1,{\bf x}_2) &=& \sum_{{\bf y}_2 \in \mathcal{Y}_2^n} W^n({\bf y}_1,{\bf y}_2|{\bf x}_1,{\bf x}_2),\\
W_2^n({\bf y}_2|{\bf x}_1,{\bf x}_2) &=& \sum_{{\bf y}_1 \in \mathcal{Y}_1^n} W^n({\bf y}_1,{\bf y}_2|{\bf x}_1,{\bf x}_2).
\end{eqnarray}
\begin{definition}
An $(n,M_n^{(1)},M_n^{(2)},\varepsilon_n^{(1)}, \varepsilon_n^{(2)})$ code for the interference channel ${\bf W}$ consists of the following essentials:
\begin{itemize}
  \item [a)]message sets:
\begin{eqnarray}
\mathcal{M}_n^{(1)} = \{1,2,\ldots,M_n^{(1)}\},\,\,\,&{\rm for\,\,\,Sender~1}\nonumber\\
\mathcal{M}_n^{(2)} = \{1,2,\ldots,M_n^{(2)}\},\,\,\,&{\rm for\,\,\,Sender~2}\nonumber
\end{eqnarray}

  \item [b)]sets of codewords:
\begin{equation}\nonumber
\begin{array}{ll}
\{{\bf x}_1(1),{\bf x}_1(2),\ldots,{\bf x}_1(M_n^{(1)})\} \subseteq \mathcal{X}_1^n,\,\,\,&{\rm for\,\,\,Encoder~1}\\
\{{\bf x}_2(1),{\bf x}_2(2),\ldots,{\bf x}_2(M_n^{(2)})\} \subseteq \mathcal{X}_2^n,\,\,\,&{\rm for\,\,\,Encoder~2}
\end{array}
\end{equation}

For Sender~1 to transmit message $i$, Encoder~1 outputs the codeword ${\bf x}_1(i)$. Similarly, for Sender~2 to transmit message $j$, Encoder~2 outputs the codeword ${\bf x}_2(j)$.

  \item [c)]collections of decoding sets:
\begin{equation}\nonumber
\begin{array}{ll}
\mathcal{B}_1 = \{\mathcal{B}_{1i} \subseteq \mathcal{Y}_1^n\}_{i = 1,...,M_n^{(1)}},\,\,&{\rm for\,\,\,Decoder~1}\\
\mathcal{B}_2 = \{\mathcal{B}_{2j} \subseteq \mathcal{Y}_2^n\}_{j = 1,...,M_n^{(2)}},\,\,&{\rm for\,\,\,Decoder~2}
\end{array}
\end{equation}
where $\mathcal{Y}_1^n = \bigcup\limits_{i = 1}^{M_n^{(1)}}\mathcal{B}_{1i},\,\,\mathcal{B}_{1i} \bigcap \mathcal{B}_{1i'} = \emptyset$ for $i \neq i'$ and $\mathcal{Y}_2^n = \bigcup\limits_{j = 1}^{M_n^{(2)}}\mathcal{B}_{2j},\,\,\mathcal{B}_{2j} \bigcap \mathcal{B}_{2j'} = \emptyset$ for $j \neq j'$.
That is, $\mathcal{B}_1$ and $\mathcal{B}_2$ are the disjoint partitions of $\mathcal{Y}_1^n$ and $\mathcal{Y}_2^n$ determined in advance, respectively.
After receiving ${\bf y}_1$, Decoder~1 outputs $\hat i$ whenever ${\bf y}_1 \in \mathcal{B}_{1\hat{i}}$. Similarly, after receiving ${\bf y}_2$, Decoder~2 outputs $\hat j$ whenever ${\bf y}_2 \in \mathcal{B}_{2\hat{j}}$.

  \item [d)]probabilities of decoding errors:
\begin{equation}\nonumber
\begin{array}{ll}
      & \varepsilon_n^{(1)} = \frac{1}{M_n^{(1)}M_n^{(2)}}\sum\limits_{i=1}^{M_n^{(1)}} \sum\limits_{j=1}^{M_n^{(2)}}W_1^n(\mathcal{B}^c_{1i}|{\bf x}_1(i), {\bf x}_2(j)),\\
      & \varepsilon_n^{(2)} = \frac{1}{M_n^{(1)}M_n^{(2)}}\sum\limits_{i=1}^{M_n^{(1)}} \sum\limits_{j=1}^{M_n^{(2)}}W_2^n(\mathcal{B}^c_{2j}|{\bf x}_1(i), {\bf x}_2(j)),
\end{array}
\end{equation}
where $``c"$ denotes the complement of a set. Here we have assumed that each message of $i \in \mathcal{M}_n^{(1)}$ and $j \in \mathcal{M}_n^{(2)}$ is produced independently with uniform distribution.
\end{itemize}
%\mathcal{X}_a^n &=& \{{\bf x}_1(1),\ldots,{\bf x}_1(M_n^{(a)})\}
\end{definition}
%{\rm Encoders} & \phi_n^{(1)}:~i \mapsto {\bf x}_1(i){\,\,\,\rm for \,\,\,}i \in \mathcal{M}_n^{(1)} \,\,\,{\rm and} \,\,\, \phi_n^{(2)}:~j \mapsto {\bf x}_2(j){\,\,\,\rm for \,\,\,}j \in \mathcal{M}_n^{(2)} \\
%      {\rm Decoders} & \psi_n^{(1)}:~{\bf y}_1 \mapsto i{\,\,\,\rm if \,\,\,}{\bf y}_1 \in \mathcal{B}_{1i} \,\,\,{\rm and} \,\,\, \psi_n^{(2)}:~{\bf y}_2 \mapsto j{\,\,\,\rm if \,\,\,}{\bf y}_2 \in \mathcal{B}_{2j}\\
{\bf Remark}: The optimal decoding to minimize the probability of errors is defining the decoding sets $\mathcal{B}_{1i}$ and $\mathcal{B}_{2j}$ according to the the maximum likelihood decoding~\cite{Etkin10}. That is, the two receivers choose, respectively,
$$\hat{i} = \arg\max_i {\rm Pr}\{{\bf y}_1|{\bf x}_1(i)\}$$
and
$$\hat{j} = \arg\max_j {\rm Pr}\{{\bf y}_2|{\bf x}_2(j)\}$$
as the estimates of the transmitted messages.
%However, the maximum likelihood decoding is difficult to analyze. Therefore, we will use information spectrum approach described in the next section.

\begin{definition}
A rate pair $(R_1,R_2)$ is achievable if there exists a sequence of $(n,M_n^{(1)},M_n^{(2)},\varepsilon_n^{(1)}, \varepsilon_n^{(2)})$ codes such that
\begin{eqnarray}
\lim_{n \rightarrow \infty} \varepsilon_n^{(1)} = 0 &{\rm and}& \lim_{n \rightarrow \infty} \varepsilon_n^{(2)} = 0 ,\nonumber\\
\liminf_{n \rightarrow \infty} \frac{\log M_n^{(1)}}{n} \geq R_1 &{\rm and}& \liminf_{n \rightarrow \infty} \frac{\log M_n^{(2)}}{n} \geq R_2.\nonumber
\end{eqnarray}
\end{definition}

\begin{definition}
The set of all achievable rates is called the capacity region of the interference channel ${\bf W}$, which is denoted by $\mathcal{C}({\bf W})$.
\end{definition}

\subsection{Preliminaries of Information-Spectrum Approach}
The following notions can be found in~\cite{Han03}.
\begin{definition}[liminf in probability]
For a sequence of random variables $\{Z^n\}_{n=1}^{\infty}$,
$$p\textrm{-}\liminf_{n \rightarrow \infty}Z^n \overset{\triangle}{=} \sup\{\beta|\lim_{n \rightarrow \infty}{\rm Pr}\{Z^n < \beta\} = 0\}.$$
\end{definition}

\begin{definition}
If two random variables sequences ${\bf X}_1 = \{{X}_1^n\}_{n=1}^{\infty}$ and ${\bf X}_2 = \{{X}_2^n\}_{n=1}^{\infty}$ satisfy that
\begin{equation}\label{P_independ}
P_{{X}_1^n{X}_2^n}({\bf x}_1,{\bf x}_2) = P_{{X}_1^n}({\bf x}_1) P_{{X}_2^n}
({\bf x}_2)
\end{equation}
for all ${\bf x}_1 \in \mathcal{X}_1^n$, ${\bf x}_2 \in \mathcal{X}_2^n$ and $n$, they are called independent and denoted by ${\bf X}_1 \bot {\bf X}_2$.
\end{definition}

Similar to~\cite{Han93}, we have
\begin{definition}
Let $S_I \stackrel{\triangle}{=} \{({\bf X}_1,{\bf X}_2)| {\bf X}_1 \bot {\bf X}_2\}$. Given an $({\bf X}_1,{\bf X}_2) \in S_I$, for the interference channel ${\bf W}$, we define the {\em spectral inf-mutual information rate} by
\begin{eqnarray}
\underline{I}({\bf X}_1;{\bf Y}_1) &\equiv& p\textrm{-}\liminf_{n \rightarrow \infty}\frac{1}{n} \log \frac{P_{Y_1^n|X_1^n}({Y}_1^n | {X}_1^n)}{P_{{Y}_1^n}({Y}_1^n)},\\
\underline{I}({\bf X}_2;{\bf Y}_2) &\equiv& p\textrm{-}\liminf_{n \rightarrow \infty}\frac{1}{n} \log \frac{P_{Y_2^n|X_2^n}({Y}_2^n | {X}_2^n)}{P_{{Y}_2^n}({Y}_2^n)},
\end{eqnarray}
where
\begin{eqnarray}
P_{Y_1^n|X_1^n}({\bf y}_1|{\bf x}_1) &=& \sum_{{\bf x}_2 ,{\bf y}_2} P_{X_2^n}({\bf x}_2)W^n({\bf y}_1,{\bf y}_2|{\bf x}_1,{\bf x}_2),\label{Py1x1}\\
P_{Y_2^n|X_2^n}({\bf y}_2|{\bf x}_2) &=& \sum_{{\bf x}_1 ,{\bf y}_1} P_{X_1^n}({\bf x}_1)W^n({\bf y}_1,{\bf y}_2|{\bf x}_1,{\bf x}_2).\label{Py2x2}
\end{eqnarray}
\end{definition}

\section{The Capacity Region of General IC}

In this section, we derive a formula for the capacity region $\mathcal{C}({\bf W})$ of the general IC.

\subsection{The Main Theorem}\label{Capacity_section}

\begin{theorem}\label{Capacity_Theorem}
The capacity region $\mathcal{C}({\bf W})$ of the interference channel ${\bf W}$ is given by
\begin{equation}
\mathcal{C}({\bf W}) = \bigcup_{({\bf X}_1,{\bf X}_2) \in S_I} \mathcal{R}_{\bf W}({\bf X}_1,{\bf X}_2),
\end{equation}
where $\mathcal{R}_{\bf W}({\bf X}_1,{\bf X}_2)$ is defined as the collection of all $(R_1,R_2)$ satisfying that
\begin{eqnarray}
0 \leq R_1 &\leq& \underline{I}({\bf X}_1;{\bf Y}_1)\label{Capacity_R1},\\
0 \leq R_2 &\leq& \underline{I}({\bf X}_2;{\bf Y}_2)\label{Capacity_R2}.
\end{eqnarray}
\end{theorem}

To prove Theorem~\ref{Capacity_Theorem}, we need the following lemmas.

\begin{lemma}\label{Error_lemma_1}
Let
$$({\bf X}_1 = \{{X}_1^n\}_{n=1}^{\infty},{\bf X}_2 = \{{X}_2^n\}_{n=1}^{\infty})$$
%:\mathcal{X}_1^n \times \mathcal{X}_2^n \rightarrow \mathcal{Y}_1^n \times \mathcal{Y}_2^n
be any channel input such that $({\bf X}_1,{\bf X}_2) \in S_I$. The corresponding output via an interference channel ${\bf W} = \{W^n\}$ is denoted by $({\bf Y}_1=\{{Y}_1^n\}_{n=1}^{\infty},{\bf Y}_2 = \{{Y}_2^n\}_{n=1}^{\infty})$. Then, for any fixed $M_n^{(1)}$ and $M_n^{(2)}$, there exists an $(n,M_n^{(1)},M_n^{(2)},\varepsilon_n^{(1)}, \varepsilon_n^{(2)})$ code satisfying that
\begin{equation}\label{Lamma_1_inequa}
\varepsilon_n^{(1)} + \varepsilon_n^{(2)} \leq {\rm Pr}\{T^c_n(1)\} + {\rm Pr}\{T^c_n(2)\} + 2e^{-n\gamma},
%\begin{array}{l}
%\varepsilon_n^{(1)} + \varepsilon_n^{(2)} \leq Pr\{\frac{1}{n} \log \frac{P_{Y_1^n|X_1^n}({Y}_1^n | {X}_1^n)}{P_{{Y}_1^n}({Y}_1^n)} \leq \frac{1}{n} \log M_n^{(1)} + \gamma\}\\
%+ Pr\{\frac{1}{n} \log \frac{P_{Y_2^n|X_2^n}({Y}_2^n | {X}_2^n)}{P_{{Y}_2^n}({Y}_2^n)} \leq \frac{1}{n} \log M_n^{(2)} + \gamma\} + 2e^{-n\gamma},
%\end{array}
\end{equation}
where
\begin{equation}
\begin{array}{l}
T_n(1) = \{({\bf x}_1,{\bf y}_1)| \frac{1}{n} \log \frac{P_{Y_1^n|X_1^n}({\bf y}_1 | {\bf x}_1)}{P_{{Y}_1^n}({\bf y}_1)} > \frac{1}{n} \log M_n^{(1)} + \gamma\},\nonumber\\
T_n(2) = \{({\bf x}_2,{\bf y}_2)| \frac{1}{n} \log \frac{P_{Y_2^n|X_2^n}({\bf y}_2 | {\bf x}_2)}{P_{{Y}_2^n}({\bf y}_2)} > \frac{1}{n} \log M_n^{(2)} + \gamma\}\nonumber
\end{array}
\end{equation}
and $\gamma > 0$ is an arbitrarily small number.
\end{lemma}
\begin{proof}[Proof of Lemma~\ref{Error_lemma_1}]
The proof is similar to that of~\cite[Lemma~3]{Han93}.

{\bf Codebook generation}. Generate $M_n^{(1)}$ independent codewords ${\bf x}_1(1),...,{\bf x}_1(M_n^{(1)}) \in \mathcal{X}_1^n$ subject to the probability distribution $P_{X_1^n}$. Similarly, generate $M_n^{(2)}$ independent codewords ${\bf x}_2(1),...,{\bf x}_2(M_n^{(2)}) \in \mathcal{X}_2^n$ subject to the probability distribution $P_{X_2^n}$.

{\bf Encoding.} To send message $i$, Sender~1 sends the codeword ${\bf x}_1(i)$. Similarly, to send message $j$, Sender~2 sends ${\bf x}_2(j)$.

{\bf Decoding.} Receiver~1 chooses the $i$ such that $({\bf x}_1(i), {\bf y}_1) \in T_n(1)$ if such $i$ exists and is unique. Similarly, Receiver~2 chooses the $j$ such that $({\bf x}_2(j), {\bf y}_2) \in T_n(2)$ if such $j$ exists and is unique. Otherwise, an error is declared.

 {\bf Analysis of the error probability.} By the symmetry of the random code construction, we can assume that $(1,1)$ was sent. Define
 $$E_{1i} = \{({\bf x}_1(i), {\bf y}_1) \in T_n(1)\},\,\, E_{2j} = \{({\bf x}_2(j), {\bf y}_2) \in T_n(2)\}.$$
For Receiver~1, an error occurs if $({\bf x}_1(1), {\bf y}_1) \notin T_n(1)$ or $({\bf x}_1(i), {\bf y}_1) \in T_n(1)$ for some $i \neq 1$.
 Similarly, for Receiver~2, an error occurs if $({\bf x}_2(1), {\bf y}_2) \notin T_n(2)$ or $({\bf x}_2(j), {\bf y}_2) \in T_n(2)$ for some $j \neq 1$. So the ensemble average of the error probabilities of Decoder~1 and Decoder~2 can be upper-bounded as follows:
 \begin{equation}\nonumber
 \begin{array}{l}
 \overline{{\varepsilon}_n^{(1)}+{\varepsilon}_n^{(2)}} = \overline{\varepsilon_n^{(1)}} + \overline{\varepsilon_n^{(2)}}\\
 \leq {\rm Pr}\{E_{11}^c\} + {\rm Pr}\{\bigcup\limits_{i \neq 1} E_{1i}\} + {\rm Pr}\{E_{21}^c\} + {\rm Pr}\{\bigcup\limits_{j \neq 1} E_{2j}\}.
 \end{array}
 \end{equation}
It can be seen that
%\begin{eqnarray}
% Pr\{E_{11}^c\} &=& Pr\{({\bf x}_1(1), {\bf y}_1) \notin T_n(1)\},\\
%  Pr\{E_{21}^c\} &=& Pr\{({\bf x}_2(1), {\bf y}_2) \notin T_n(2)\},
%\end{eqnarray}
%and
  \begin{equation}\nonumber
 \begin{array}{ll}
 &{\rm Pr}\{\bigcup\limits_{i \neq 1} E_{1i}\} \leq \sum\limits_{i \neq 1}{\rm Pr}\{E_{1i}\}\\
 &=\sum\limits_{i \neq 1} {\rm Pr}\{({\bf x}_1(i), {\bf y}_1) \in T_n(1)\}\\
&\stackrel{(a)}{=} \sum\limits_{i \neq 1} \sum\limits_{({\bf x}_1,{\bf y}_1) \in T_n(1)}P_{X_1^n}({\bf x}_1)P_{Y_1^n}({\bf y}_1)\\
&\stackrel{(b)}{\leq} \sum\limits_{i \neq 1} \sum\limits_{({\bf x}_1,{\bf y}_1) \in T_n(1)}P_{X_1^n}({\bf x}_1)P_{Y_1^n|X_1^n}({\bf y}_1|{\bf x}_1)\frac{e^{-n\gamma}}{M_n^{(1)}}\\
&\leq  \sum\limits_{i \neq 1}\frac{e^{-n\gamma}}{M_n^{(1)}} = (M_n^{(1)}- 1)\frac{e^{-n\gamma}}{M_n^{(1)}} \leq e^{-n\gamma},
 \end{array}
 \end{equation}
 where $(a)$ follows from the independence of ${\bf x}_1(i)~(i \neq 1)$ and ${\bf y}_1$ and $(b)$ follows from the definition of $T_n(1)$. Similarly, we obtain
 \begin{equation}
 {\rm Pr}\{\bigcup\limits_{j \neq 1} E_{2j}\} \leq e^{-n\gamma}.
 \end{equation}
 Combining all inequalities above, we can see that there must exist at least one $(n,M_n^{(1)},M_n^{(2)},\varepsilon_n^{(1)}, \varepsilon_n^{(2)})$ code satisfying~(\ref{Lamma_1_inequa}).
\end{proof}

\begin{lemma}\label{Error_lemma_2}
For all $n$, any $(n,M_n^{(1)},M_n^{(2)},\varepsilon_n^{(1)}, \varepsilon_n^{(2)})$ code satisfies that
\begin{equation}\label{leq}
\begin{array}{l}
\varepsilon_n^{(1)} \geq {\rm Pr}\{\frac{1}{n} \log \frac{P_{Y_1^n|X_1^n}({Y}_1^n | {X}_1^n)}{P_{{Y}_1^n}({Y}_1^n)} \leq \frac{1}{n} \log M_n^{(1)} - \gamma\} - e^{-n\gamma},\\
\varepsilon_n^{(2)} \geq {\rm Pr}\{\frac{1}{n} \log \frac{P_{Y_2^n|X_2^n}({Y}_2^n | {X}_2^n)}{P_{{Y}_2^n}({Y}_2^n)} \leq \frac{1}{n} \log M_n^{(2)} - \gamma\} - e^{-n\gamma},
\end{array}
\end{equation}
for every $\gamma > 0$, where $X_1^n~({\rm resp.,}\,X_2^n)$ places probability mass $1/M_n^{(1)}~({\rm resp.,}\,1/M_n^{(2)})$ on each codeword for Encoder 1~(resp., Encoder 2) and (\ref{P_independ}), (\ref{Py1x1}), (\ref{Py2x2}) hold.
\end{lemma}

\begin{proof}[Proof of Lemma~\ref{Error_lemma_2}]
The proof is similar to that of~\cite[Lemma~4]{Han93}. By using the relation
$$ \frac{P_{Y_1^n|X_1^n}({\bf y}_1 | {\bf x}_1)}{P_{{Y}_1^n}({\bf y}_1)} =  \frac{P_{X_1^n|Y_1^n}({\bf x}_1 | {\bf y}_1)}{P_{{X}_1^n}({\bf x}_1)}$$
and noticing that $P_{{X}_1^n}({\bf x}_1) = \frac{1}{M_n^{(1)}}$, we can rewrite the first term on the right-hand side of the first inequality of (\ref{leq}) as
$${\rm Pr}\{P_{X_1^n|Y_1^n}(X_1^n|Y_1^n) \leq e^{-n\gamma}\}.$$
By setting
$$L_n = \{({\bf x}_1,{\bf y}_1)|P_{X_1^n|Y_1^n}({\bf x}_1 | {\bf y}_1) \leq  e^{-n\gamma}\},$$
the first inequality of (\ref{leq}) can be expressed as
\begin{equation}
{\rm Pr}\{L_n\} \leq \varepsilon_n^{(1)} +  e^{-n\gamma}.
\end{equation}
In order to prove this inequality, we set
$$\mathcal{A}_i = \{{\bf y}_1 \in \mathcal{Y}_1^n| P_{X_1^n|Y_1^n}({\bf x}_1(i) | {\bf y}_1) \leq  e^{-n\gamma}\}.$$
It can be seen that
\begin{equation}\nonumber
\begin{array}{l}
{\rm Pr}\{L_n\} = \sum\limits_{i = 1}^{M_n^{(1)}} P_{X_1^nY_1^n}({\bf x}_1(i),\mathcal{A}_i)\\
= \sum\limits_{i = 1}^{M_n^{(1)}} P_{X_1^nY_1^n}({\bf x}_1(i),\mathcal{A}_i\bigcap \mathcal{B}_{1i}) + \sum\limits_{i = 1}^{M_n^{(1)}} P_{X_1^nY_1^n}({\bf x}_1(i),\mathcal{A}_i\bigcap \mathcal{B}_{1i}^c)\\
\leq \sum\limits_{i = 1}^{M_n^{(1)}} P_{X_1^nY_1^n}({\bf x}_1(i),\mathcal{A}_i\bigcap \mathcal{B}_{1i}) + \sum\limits_{i = 1}^{M_n^{(1)}} P_{X_1^nY_1^n}({\bf x}_1(i),\mathcal{B}_{1i}^c)\\
= \sum\limits_{i = 1}^{M_n^{(1)}}\sum\limits_{{\bf y}_1 \in \mathcal{A}_i\bigcap \mathcal{B}_{1i}} P_{X_1^nY_1^n}({\bf x}_1(i),{\bf y}_1) + \varepsilon_n^{(1)} \\
=\sum\limits_{i = 1}^{M_n^{(1)}}\sum\limits_{{\bf y}_1 \in \mathcal{A}_i\bigcap \mathcal{B}_{1i}} P_{X_1^n|Y_1^n}({\bf x}_1(i)|{\bf y}_1)P_{Y_1^n}({\bf y}_1) + \varepsilon_n^{(1)} \\
\stackrel{(a)}{\leq} e^{-n\gamma}\sum\limits_{i = 1}^{M_n^{(1)}}\sum\limits_{{\bf y}_1 \in \mathcal{B}_{1i}}P_{Y_1^n}({\bf y}_1)+ \varepsilon_n^{(1)} \\
= e^{-n\gamma}P_{Y_1^n}(\bigcup\limits_{i = 1}^{M_n^{(1)}}\mathcal{B}_{1i})+ \varepsilon_n^{(1)} \leq e^{-n\gamma} + \varepsilon_n^{(1)},
\end{array}
\end{equation}
where $\mathcal{B}_{1i}$ is the decoding region corresponding to codeword ${\bf x}_1(i)$ and $(a)$ follows from the definition of $\mathcal{A}_i$. Therefore, the first inequality of (\ref{leq}) is proved. Similarly, we can obtain the second inequality of (\ref{leq}).
\end{proof}

Now we prove Theorem~\ref{Capacity_Theorem}.
\begin{proof}[Proof of Theorem~\ref{Capacity_Theorem}]

1)~To prove that an arbitrary rate pair $(R_1,R_2)$ satisfying (\ref{Capacity_R1}) and (\ref{Capacity_R2}) is achievable, we define
$$M_n^{(1)} = e^{n(R_1-2\gamma)}\,\,\,{\rm and}\,\,\,M_n^{(2)} = e^{n(R_2-2\gamma)}$$
for an arbitrarily small constant $\gamma > 0$. Lemma~\ref{Error_lemma_1} guarantees the existence of an $(n,M_n^{(1)},M_n^{(2)},\varepsilon_n^{(1)}, \varepsilon_n^{(2)})$ code satisfying

\begin{equation}\nonumber\label{Direct_leq1}
\begin{array}{ll}
\varepsilon_n^{(1)} + \varepsilon_n^{(2)} &\leq {\rm Pr}\{\frac{1}{n} \log \frac{P_{Y_1^n|X_1^n}({Y}_1^n | {X}_1^n)}{P_{{Y}_1^n}({Y}_1^n)} \leq  R_1 - \gamma\}\\
&\,\,\,\,\,\,+{\rm Pr}\{\frac{1}{n} \log \frac{P_{Y_2^n|X_2^n}({Y}_2^n | {X}_2^n)}{P_{{Y}_2^n}({Y}_2^n)} \leq  R_2 - \gamma\} + 2e^{-n\gamma}\\
 &\leq {\rm Pr}\{\frac{1}{n} \log \frac{P_{Y_1^n|X_1^n}({Y}_1^n | {X}_1^n)}{P_{{Y}_1^n}({Y}_1^n)} \leq  \underline{I}({\bf X}_1;{\bf Y}_1)- \gamma\} \\
 &\,\,\,\,\,\,+ {\rm Pr}\{\frac{1}{n} \log \frac{P_{Y_2^n|X_2^n}({Y}_2^n | {X}_2^n)}{P_{{Y}_2^n}({Y}_2^n)} \leq  \underline{I}({\bf X}_2;{\bf Y}_2)- \gamma\} + 2e^{-n\gamma}.
\end{array}
\end{equation}
From the definition of the spectral inf-mutual information rate, we have
$$\lim_{n \rightarrow \infty} \varepsilon_n^{(1)} = 0 \,\,\,{\rm and}\,\,\,\lim_{n \rightarrow \infty} \varepsilon_n^{(2)} = 0.$$

2)~Suppose that a rate pair $(R_1,R_2)$ is achievable. Then, for any constant $\gamma > 0$, there exists an $(n,M_n^{(1)},M_n^{(2)},\varepsilon_n^{(1)}, \varepsilon_n^{(2)})$ code satisfying
\begin{equation}
\frac{\log M_n^{(1)}}{n} \geq R_1 - \gamma \,\,\,{\rm and}\,\,\, \frac{\log M_n^{(2)}}{n} \geq R_2 - \gamma
\end{equation}
for all sufficiently large $n$ and
$$\lim_{n \rightarrow \infty} \varepsilon_n^{(1)} = 0 \,\,\,{\rm and}\,\,\,\lim_{n \rightarrow \infty} \varepsilon_n^{(2)} = 0.$$
From Lemma~\ref{Error_lemma_2}, we get
\begin{equation}\label{geqR}
\begin{array}{l}
\varepsilon_n^{(1)} \geq {\rm Pr}\{\frac{1}{n} \log \frac{P_{Y_1^n|X_1^n}({Y}_1^n | {X}_1^n)}{P_{{Y}_1^n}({Y}_1^n)} \leq R_1 - 2\gamma\} - e^{-n\gamma}\\
\varepsilon_n^{(2)} \geq {\rm Pr}\{\frac{1}{n} \log \frac{P_{Y_2^n|X_2^n}({Y}_2^n | {X}_2^n)}{P_{{Y}_2^n}({Y}_2^n)} \leq R_2 - 2\gamma\} - e^{-n\gamma}
\end{array}.
\end{equation}
Taking the limits as $n\rightarrow \infty$ on both sides, we have
\begin{equation}\label{geqR1}
\begin{array}{l}
\lim\limits_{n\rightarrow \infty}{\rm Pr}\{\frac{1}{n} \log \frac{P_{Y_1^n|X_1^n}({Y}_1^n | {X}_1^n)}{P_{{Y}_1^n}({Y}_1^n)} \leq R_1 - 2\gamma\} = 0\\
\lim\limits_{n\rightarrow \infty}{\rm Pr}\{\frac{1}{n} \log \frac{P_{Y_2^n|X_2^n}({Y}_2^n | {X}_2^n)}{P_{{Y}_2^n}({Y}_2^n)} \leq R_2 - 2\gamma\} = 0
\end{array}.
\end{equation}
From the definitions of $\underline{I}({\bf X}_1;{\bf Y}_1)$ and $\underline{I}({\bf X}_2;{\bf Y}_2)$, we can see that $R_1 - 2\gamma \leq \underline{I}({\bf X}_1;{\bf Y}_1)$ and $R_2 - 2\gamma \leq \underline{I}({\bf X}_2;{\bf Y}_2)$, which completes the proof since $\gamma$ is arbitrary.

%Implying by the definitions of , it can be seen that . This completes the proof since  is arbitrary.
%
%Assume that (\ref{Capacity_R1}) does not hold. Then, we can choose a sufficiently small $\gamma > 0$ satisfying
%$$R_1 - 3\gamma > \underline{I}({\bf X}_1;{\bf Y}_1).$$
%From the first inequality of (\ref{leq}) it can be seen that
%\begin{equation}\label{contradicts}
%\varepsilon_n^{(1)} \geq {\rm Pr}\{\frac{1}{n} \log \frac{P_{Y_1^n|X_1^n}({Y}_1^n | {X}_1^n)}{P_{{Y}_1^n}({Y}_1^n)} \leq \underline{I}({\bf X}_1;{\bf Y}_1)+ \gamma\} - e^{-n\gamma}.
%\end{equation}
%However, since $\varepsilon_n^{(1)} \rightarrow 0$ and $e^{-n\gamma} \rightarrow 0$ as $n \rightarrow \infty$, (\ref{contradicts}) contradicts the definition of {\em spectral inf-mutual information rate} $\underline{I}({\bf X}_1;{\bf Y}_1)$. This means that (\ref{Capacity_R1}) must hold. Similarly, we can prove (\ref{Capacity_R2}) by using the second inequality of (\ref{geqR}).
\end{proof}

\subsection{The Algorithm to Compute Achievable Rate Pairs}\label{Sec:algFortheorem}
Theorem~\ref{Capacity_Theorem} provides a general formula for the capacity region of a general IC. However, it is usually difficult to compute the spectral inf-mutual information rates given in (\ref{Capacity_R1}) and (\ref{Capacity_R2}). In order to get insights into the interference channels, we make the following assumptions:
\begin{itemize}
  \item [1)] the channel is stationary and memoryless, that is, the transition probability of the channel can be written as
$$W^n({\bf y}_1,{\bf y}_2|{\bf x}_1,{\bf x}_2) = \prod_{i = 1}^nW(y_{1,i},y_{2,i}|x_{1,i},x_{2,i});$$
  \item [2)]sources are restricted to be stationary and ergodic discrete Markov processes.
\end{itemize}
With the above assumptions, the spectral inf-mutual information rates are reduced as
\begin{eqnarray}
\underline{I}({\bf X}_1;{\bf Y}_1) &=& \lim_{n \rightarrow \infty}\frac{1}{n}I({X}_1^n;{Y}_1^n),\\
\underline{I}({\bf X}_2;{\bf Y}_2) &=& \lim_{n \rightarrow \infty}\frac{1}{n}I({X}_2^n;{Y}_2^n),
\end{eqnarray}
which can be evaluated by the Monte Carlo method~\cite{Kavcic01}\cite{Arnold01}\cite{Pfister01} using BCJR algorithm~\cite{BCJR74} over a trellis. Actually, any stationary and ergodic discrete Markov source can be represented by a time-invariant trellis and (hence) is uniquely specified by a trellis section. A trellis section is composed of left~(or starting) states and right~(or ending) states, which are connected by branches in between. For example, Source ${\bf x}_1$ can be specified by a trellis $\mathcal{T}_1$ as follows.
\begin{itemize}
  \item Both the left and right states are selected from the set $\mathcal{S}_1 = \{0,1,...,|\mathcal{S}_1| - 1\}$;
 % At time $t-1$, the source ${\bf x}_1$ in the state $s_{1,t-1}$ emits a letter $x_{1,t} \in \mathcal{X}_1$ and goes into a new state $s_{1,t}$.
  \item Each branch is represented by a three-tuple $b = (s_1^-(b),x_1(b), s_1^{+}(b))$, where $s_1^{-}(b)$ is the left state, $s_1^{+}(b)$
  is the right state, and the symbol $x_1(b) \in \mathcal{X}_1$ is the associated label.
  %generated according to the conditional probability $P(x_1(b)|s_1^-(b))$.
  We also assume that a branch $b$ is uniquely determined by $s_1^-(b)$ and $x_1(b)$;

   \item At time $t = 0$, the source starts from state $s_{1,0}\in \mathcal{S}_1$. If at time $t-1~(t>0)$, the source is in the state $s_{1,t-1} \in \mathcal{S}_1$,
    then at time $t~(t>0)$, the source generates a symbol $x_{1,t} \in \mathcal{X}_1$ according to the conditional probability $P(x_{1,t}|s_{1,t-1})$ and goes into a state $s_{1,t}  \in \mathcal{S}_1$
    such that $(s_{1,t-1} ,x_{1,t},s_{1,t} )$ is a branch. Obviously, when the source runs from time $t=0$ to $t = n$, a sequence $x_{1,1},x_{1,2},...,x_{1,n}$ is generated. The Markov property says that
    $$P(x_{1,t}|x_{1,1},...,x_{1,{t-1}},s_{1,0} ) = P(x_{1,t}|s_{1,t-1} ).$$
    So the probability of a given sequence $x_{1,1},x_{1,2},...,x_{1,n}$ with the initial state $s_{1,0} $ can be factored as
    $$P(x_{1,1},x_{1,2},...,x_{1,n} | s_{1,0} ) = \prod\limits_{t = 1}^n  P(x_{1,t}|s_{1,t-1} ).$$

  %  the output is a path through the trellis from time $0$ to time $N$, denoted by $\underline{p}_{1,0}^N$. The Markov property says that the probability $P(b_1^t|\underline{p}_{1,0}^{t-1})$ of generating a branch $b_1^t$ at time $t$ given the generated path up to time $t-1$, depends only on the ending state of the generated path
%    $$P(b_1^t|\underline{p}_{1,0}^{t-1}) = P(b_1^t | s_1^-(b_1^t)).$$
%    So, the probability of a given path $\underline{p}_{1,0}^{N} = b_1^0 \rightarrow  b_1^1 \rightarrow ... \rightarrow b_1^N$ factors as
%    $$P(\underline{p}_{1,0}^{N}) = P(s_1^-(b_1^0))\prod\limits_{t = 1}^N P(b_1^t|s_1^-(b_1^t)),$$
%    where $P(b_1^t|s_1^-(b_1^t))$ is assumed to be known for all $t$.

\end{itemize}

Similarly, we can represent ${\bf x}_2$ by a trellis $\mathcal{T}_2$ with the state set $\mathcal{S}_2 = \{0,1,...,|\mathcal{S}_2| - 1\}$. Each branch is denoted by $b = (s_2^-(b),x_2(b), s_2^{+}(b))$, where $s_2^{-}(b)$ is the left state, $s_2^{+}(b)$
  is the right state and the symbol $x_2(b) \in \mathcal{X}_2$ is  the associated label. Assume that source ${\bf x}_2$ starts from the state $s_{2,0} \in \mathcal{S}_2$. If at time $t-1~(t>0)$, the source is in the state $s_{2,t-1}\in \mathcal{S}_2$,
    then at time $t~(t>0)$, the source generates a symbol $x_{2,t} \in \mathcal{X}_2$ according to the conditional probability $P(x_{2,t}|s_{2,t-1})$ and goes into a state $s_{2,t} \in \mathcal{S}_2$
    such that $(s_{2,t-1},x_{2,t},s_{2,t})$ is a branch. The probability of a given sequence $x_{2,1},x_{2,2},...,x_{2,n}$ can be factored as
    $$P(x_{2,1},x_{2,2},...,x_{2,n} | s_{2,0}) = \prod\limits_{t = 1}^n P(x_{2,t}|s_{2,t-1}).$$
In what follows, we have fixed
the initial states as $s_{1,0} = 0$ and $s_{2,0} = 0$, and removed them from the equations for simplicity.

Next we focus on the evaluation of $\lim\limits_{n \rightarrow \infty}\frac{1}{n}I({X}_1^n;{Y}_1^n)$, while $\lim\limits_{n \rightarrow \infty}\frac{1}{n}I({X}_2^n;{Y}_2^n)$ can be estimated similarly. Specifically, we can express the limit as
\begin{equation}\label{eq:MItoEntropy}
  \lim_{n \rightarrow \infty}\frac{1}{n}I({X}_1^n;{Y}_1^n) = \lim_{n \rightarrow \infty}\frac{1}{n}H(Y_1^n) - \lim_{n \rightarrow \infty}\frac{1}{n}H(Y_1^n | X_1^n),
\end{equation}
where $\lim\limits_{n \rightarrow \infty}\frac{1}{n}H(Y_1^n)$ and $\lim\limits_{n \rightarrow \infty}\frac{1}{n}H(Y_1^n | X_1^n)$ can be estimated by similar methods\footnote{For continuous ${\bf y}_1$, the computations of  $\lim\limits_{n \rightarrow \infty}\frac{1}{n}h(Y_1^n)$ and $ \lim\limits_{n \rightarrow \infty}\frac{1}{n}h(Y_1^n | X_1^n)$ can be implemented by substituting pdf for pmf.}. As an example, we show how to compute $\lim\limits_{n \rightarrow \infty}\frac{1}{n}H(Y_1^n)$.
According to the Shannon-McMillan-Breiman theorem~\cite{Cover91}, it can be seen that, with probability 1,
$$\lim_{n \rightarrow \infty}-\frac{1}{n}\log P(y_{1}^n) = \lim_{n \rightarrow \infty}\frac{1}{n}H(Y_1^n),$$
 where $y_{1}^n$ stands for $(y_{1,1},y_{1,2},...,y_{1,n})$.
Then evaluating $\lim\limits_{n \rightarrow \infty}\frac{1}{n}H(Y_1^n)$ is
converted to computing
$$\lim_{n \rightarrow \infty}-\frac{1}{n}\log P(y_{1}^n) \approx -\frac{1}{n}\log\left(\prod_{t = 1}^n P(y_
{1,t}|y_{1}^{t-1})\right) = -\frac{1}{n}\sum_{t = 1}^n\log P(y_
{1,t}|y_{1}^{t-1})$$
for a sufficiently long typical sequence $y_{1}^n$. Here, the key is to compute the conditional probabilities $P(y_
{1,t}|y_{1}^{t-1})$ for all $t$. Since both ${\bf y}_1$ and ${\bf y}_2$ are hidden Markov sequences, this can be done by performing the BCJR algorithm over the following product trellis.

\begin{itemize}
  \item The product trellis has the state set $\mathcal{S} = \mathcal{S}_1 \times \mathcal{S}_2$, where ``$\times$" denotes Cartesian product.

  \item Each branch is represented by a four-tuple $b = (s^-(b),x_1(b),x_2(b),s^{+}(b))$, where $s^-(b) = (s_1^-(b),s_2^-(b))$ is the left state,
  $s^+(b) = (s_1^+(b),s_2^+(b))$ is the right state. Then $x_1(b) \in \mathcal{X}_1$ and $x_2(b) \in \mathcal{X}_2$ are the associated labels in branch~$b$ such
  that $(s_1^-(b),x_1(b),s_1^+(b))$ and $(s_2^-(b),x_2(b),s_2^+(b))$ are branches in $\mathcal{T}_1$ and $\mathcal{T}_2$, respectively.

  \item  At time $t = 0$, the sources start from state $s_0 = (s_{1,0},s_{2,0}) \in \mathcal{S}$. If at time $t-1~(t>0)$, the sources are in the state $s_{t-1} = (s_{1,t-1},s_{2,t-1}) \in \mathcal{S}$,
    then at time $t~(t>0)$, the sources generate symbols $(x_{1,t} \in \mathcal{X}_1, x_{2,t} \in \mathcal{X}_2)$ according to the conditional probability $P(x_{1,t}|s_{1,t-1})P(x_{2,t}|s_{2,t-1})$ and go into a state $s_t = (s_{1,t},s_{2,t}) \in \mathcal{S}_2$
    such that $(s_{t-1},x_{1,t},x_{2,t},s_t)$ is a branch.

 % Denote the path through the product trellis from time $0$ to time $t$ by $\underline{p}_{0}^{t-1}$, the probability $P(b^t|\underline{p}_{0}^{t-1})$ of generating a branch $b^t$ at time $t$ given the generated path up to time $t-1$ equals to
%$P(b^t|s^-(b^t))$. Therefore, the probability of a given path $\underline{p}_{0}^{N} = b^0 \rightarrow  b^1 \rightarrow ... \rightarrow b^N$ factors as
%    $$P(\underline{p}_{0}^{N}) = P(s^-(b^0))\prod\limits_{t = 1}^N P(b^t|s^-(b^t)).$$
% We denote the collection of all branches by $\mathcal{B}$.
      %So there are $|\mathcal{X}_1|$ parallel branches from a left state to a right state.
\end{itemize}

Given the received sequence ${\bf y}_1$, we define
\begin{itemize}
  \item {\em Branch metrics}: To each branch $b_t = \{s_{t-1}, x_{1,t}, x_{2,t},s_t\}$, we assign a metric
  \begin{eqnarray}
    \rho(b_t) &\overset{\triangle}{=}& P(b_t|s_{t-1}) P(y_{1,t}|x_{1,t}x_{2,t})\\
     &=& P(x_{1,t}|s_{1,t-1})P(x_{2,t}|s_{2,t-1})P(y_{1,t}|x_{1,t}x_{2,t}),
  \end{eqnarray}
In the computation of $\lim\limits_{n \rightarrow \infty}\frac{1}{n}H(Y_1^n | X_1^n)$, the metric is replaced by $P(b_t|s_{t-1},x_{1,t}) P(y_{1,t}|x_{1,t}x_{2,t})$.
  \item {\em State transition probabilities}: The transition probability from $s_{t-1}$ to $s_t$ is defined as
  \begin{eqnarray}
    \gamma_t(s_{t-1}, s_t) &\overset{\triangle}{=}& P(s_t,y_{1,t}|s_{t-1}) \\
    &=& \sum_{b_t:s^-(b_t) = s_{t-1},s^+(b_t) = s_t}\rho(b_t).
  \end{eqnarray}
  \item {\em Forward recursion variables}: We define the \emph{a posteriori} probabilities
  \begin{equation}\label{posterior probabilities}
    \alpha_t(s_t) \overset{\triangle}{=} P(s_t|y_{1}^t), \,\,\,t= 0,1,...n.
  \end{equation}
  Then
  \begin{equation}\label{eq:Py_t}
    P(y_{1,t}|y_{1}^{t-1}) = \sum_{s_{t-1},s_t} \alpha(s_{t-1})\gamma_t(s_{t-1}, s_t),
  \end{equation}
where the values of $\alpha_t(s_t)$ can be computed recursively by
\begin{equation}\label{eq:alfa_t}
\alpha_t(s_t) = \frac{\sum_{s_{t-1}} \alpha_{t-1}(s_{t-1})\gamma_t(s_{t-1}, s_t)}{\sum_{s_{t-1},s_t} \alpha_{t-1}(s_{t-1})\gamma_t(s_{t-1}, s_t)}.
\end{equation}
\end{itemize}
In summary, the algorithm to estimate the entropy rate $\lim\limits_{n \rightarrow \infty}\frac{1}{n}H(Y_1^n)$ is described as follows.

%For a general IC defined by the probability transition matrix $W^n({\bf y}_1,{\bf y}_2|{\bf x}_1,{\bf x}_2)$, whose senders take Markov processes, we give Algorithm~\ref{alg:GIC_capacity} to compute its capacity region.

\begin{algorithm}
\mbox{}\par
\begin{enumerate}
  \item {\bf Initializations:} Choose a sufficiently large number $n$. Set the initial state of the trellis to be $s_0 = 0$. The forward recursion variables are initialized as $\alpha_0(s) = 1$ if $s = 0$ and otherwise $\alpha_0(s) = 0$.
  \item {\bf Simulations for Sender~1:}
     Generate a Markov sequence ${\bf x}_1 = (x_{1,1}, x_{1,2},...,x_{1,n})$ according to the trellis $\mathcal{T}_1$ of source ${\bf x}_1$.

  \item {\bf Simulations for Sender~2:}
  Generate a Markov sequence ${\bf x}_2 = (x_{2,1}, x_{2,2},...,x_{2,n})$ according to the trellis $\mathcal{T}_2$ of source ${\bf x}_2$.

   \item {\bf Simulations for Receiver~1:}
   Generate the received sequence ${\bf y}_1$ according to the transition probability $W^n({\bf y}_1, {\bf y}_2 | {\bf x}_1, {\bf x}_2)$.
     %\begin{itemize}
%    \item [a)]Generate an $n$-sequence $N_1 \in \mathbb{R}^n$ independently according to the Gaussian
%distribution $\mathcal{N}(0, 1)$.
%  \item [b)] Receive the sequence ${\bf y}_1 = {\bf x}_1 + \sqrt{a} {\bf x}_2 + N_1$.
%
%  \end{itemize}

   \item {\bf Computations:}
      \begin{itemize}
    \item [a)]For $t = 1,2,...,n$ , compute the values of $P(y_{1,t}|y_{1}^{t-1})$ and $\alpha_t(s_t)$ recursively according to equations (\ref{eq:Py_t}) and (\ref{eq:alfa_t}).
  \item [b)] Evaluate the entropy rate
  $$ \lim_{n \rightarrow \infty}\frac{1}{n}H(Y_1^n) = -\frac{1}{n}\sum_{t = 1}^n\log P(y_{1,t}|y_{1}^{t-1}).$$
  \end{itemize}
\end{enumerate}
\end{algorithm}

Similarly, we can evaluate the entropy rate $\lim\limits_{n \rightarrow \infty}\frac{1}{n}H(Y_1^n | X_1^n)$. Therefore, we obtain the achievable rate
$\underline{I}({\bf X}_1;{\bf Y}_1) = \lim \limits_{n \rightarrow \infty}\frac{1}{n}I({X}_1^n;{Y}_1^n)$.
%\begin{itemize}
%  \item[1)] Let ${\bf x}_1$ and ${\bf x}_2$  be Markov sequences with probability mass function~(pmf) $P_{X_1^n}$ and $P_{X_2^n}$, respectively. Set $n = 1000000$.
%  \item[2)] Use BCJR algorithm to compute $\frac{1}{n}I({X}_1^n;{Y}_1^n)$ and $\frac{1}{n}I({X}_2^n;{Y}_2^n)$. Get a rectangular region by
%  \begin{eqnarray}\nonumber
%0 \leq R_1 &\leq& \underline{I}({\bf X}_1;{\bf Y}_1)\\
%0 \leq R_2 &\leq& \underline{I}({\bf X}_2;{\bf Y}_2).
%\end{eqnarray}
%
%\begin{enumerate}
%  \item 
%  \item 
%  \item 
%\end{enumerate}
%  \item[3)] Update pmf $P_{X_1^n}$ and $P_{X_2^n}$ and return to step 1). Take convex and union operations over all the rectangular regions and get the capacity region.
%\end{itemize}

\subsection{Numerical Results}\label{GIFC_section}

%We have obtained a formula of the capacity region for the general IC, which shows that any pair of independent input processes define a pair of achievable rates. Although it is not computable in general, the formula provides us useful insights into the interference channels especially when the input processes are constrained to be stationary. In this section, we take symmetric GIFC as an example to present some numerical results.
%We have obtained a formula of the capacity region for the¡¡
%general IC, which shows that any pair of independent input
%processes define a pair of achievable rates. Although it is
%not computable in general, the formula provides us useful
%insights into the interference channels. In particular, if both senders take Markov processes over certain fixed signal constellations as inputs, the achievable rates are computable, as is illustrated by the following example.

%--------------------------------------------------------------------------------------------------------
\begin{figure}
  \centering
  \includegraphics[width=11.0cm]{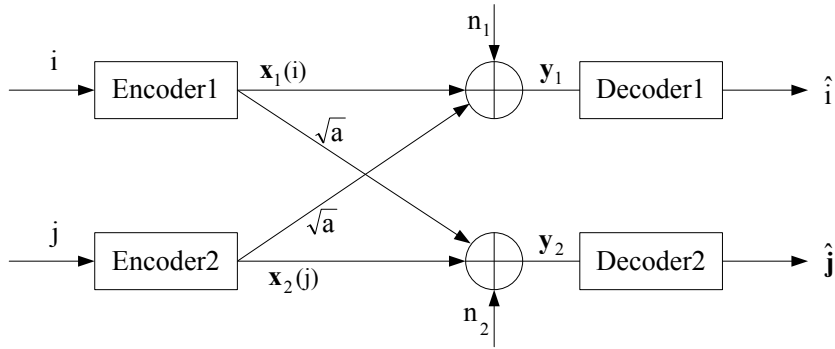}\\
  \caption{Symmetric Gaussian interference channel.}\label{GIFC}
\end{figure}
%-------------------------------------------------------------------------------------------------------
We consider the model as shown in Fig.~\ref{GIFC}, where the channel inputs
${\bf x}_1(i)$ and ${\bf x}_2(j)$ are BPSK sequences with power constraints $P_1$ and $P_2$, respectively; the additive noise ${\bf n}_1$ and ${\bf n}_2$ are sequences of independent and identically distributed~(i.i.d.) standard Gaussian random variables; the channel outputs ${\bf y}_1$ and ${\bf y}_2$ are
\begin{eqnarray}
{\bf y}_1 &=& {\bf x}_1(i) + \sqrt{a}{\bf x}_2(j) + {\bf n}_1,\\
{\bf y}_2 &=& {\bf x}_2(j) + \sqrt{a}{\bf x}_1(i) + {\bf n}_2.
\end{eqnarray}
%For any two arbitrary input processes ${\bf x}_1$ and ${\bf x}_2$, the defined pair of rates given in Theorem~\ref{Capacity_Theorem} are infeasible to compute.  Therefore,
We assume that ${\bf x}_1$ and
${\bf x}_2$ are the outputs from two~(possibly different) generalized trellis encoders driven by independent and uniformly distributed~(i.u.d.) input sequences, as proposed in~\cite{Huang2011}.
%In this case, both ${\bf x}_1$ and
%${\bf x}_2$ are block-wise stationary, and~(hence)
%\begin{eqnarray}
%\underline{I}({\bf X}_1;{\bf Y}_1) &=& \lim_{n \rightarrow \infty}\frac{1}{n}I({X}_1^n;{Y}_1^n),\label{Spectrum_information_1}\\
%\underline{I}({\bf X}_2;{\bf Y}_2) &=& \lim_{n \rightarrow \infty}\frac{1}{n}I({X}_2^n;{Y}_2^n).\label{Spectrum_information_2}
%\end{eqnarray}
%Since ${\bf x}_1, {\bf x}_2$ and ${\bf y}_1, {\bf y}_2$ can be viewed as the Markov chains and the hidden Markov chains, respectively, the right-hand sides of~(\ref{Spectrum_information_1}) and~(\ref{Spectrum_information_2}) can be estimated by performing the BCJR algorithm~\footnote{Only forward recursion is required.}~\cite{Arnold06}\cite{BCJR74}.
As examples, we consider two input processes. One is referred to as ``UnBPSK", standing for an i.u.d. BPSK sequence; the other is referred to as ``CcBPSK", standing for an output sequence from the convolutional encoder with the generator matrix
$
G(D) = [1 + D + D^2\,\,\,1+D^2]
$
driven by an i.u.d. input sequence.

Fig.~\ref{FigTRIFFC} shows the trellis representation of the signal model when Sender~1 uses CcBPSK and Sender~2 uses UnBPSK.
Fig.~\ref{UncodedVSConv_figure} shows the numerical results. There are three rectangles, OECH, ODBG and OFAI, each of which is determined by a pair of spectral inf-mutual
information rates. Specifically, the rectangle OECH corresponds to the case when both senders use UnBPSK as inputs; the rectangle ODBG corresponds to the case when Sender~1 uses UnBPSK as input and Sender~2 uses CcBPSK as input; and the rectangle OFAI corresponds to the case when Sender~1 uses CcBPSK as input and Sender~2 uses UnBPSK as input. The point ``A" can be achieved by a coding scheme, in which
Sender~1 uses a binary linear~(coset) code concatenated with the convolutional code and Sender~2 uses a binary linear code,
and the point ``B" can be achieved similarly; while the points on the line ``AB" can be achieved by time-sharing scheme.
The point ``C" represents the limits when the two senders use binary linear codes but regard the interference as an i.u.d. additive~(BPSK) noise. It can be seen that the area of the pentagonal region ODBAI is greater
than that of the rectangle OECH, which implies that knowing the structure of the interference can be used to improve potentially the bandwidth-efficiency.

%------------------------------------------------------------------------------------------------------
\begin{figure}
  \centering
  \includegraphics[width=15.0cm]{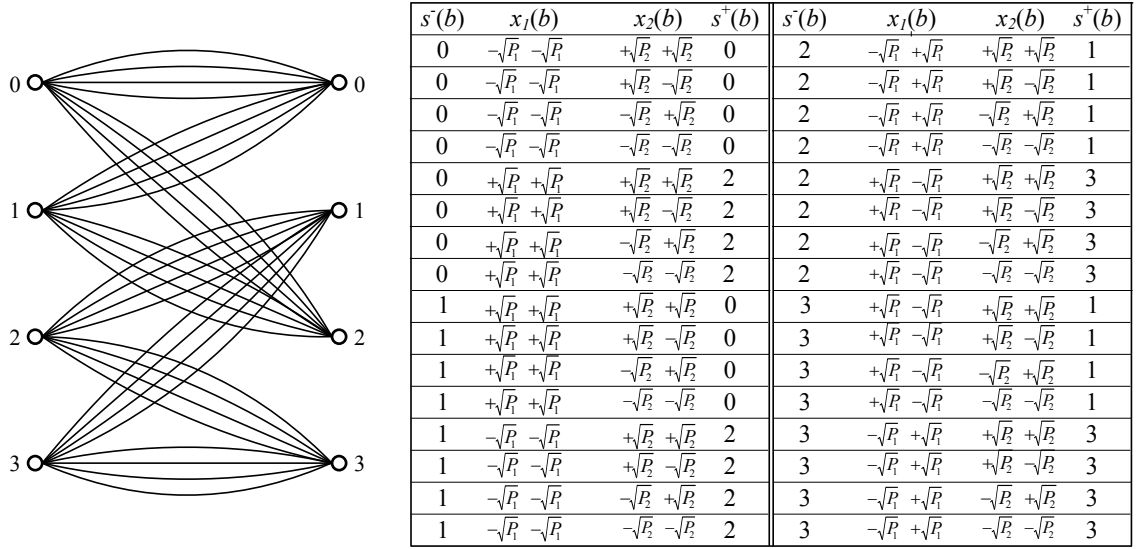}\\
  \caption{The trellis section of~(CcBPSK, UnBPSK) with 32 branches. For each branch $b$, $s^-(b)$ and $s^+(b)$ are the left state and the right state, respectively; while the associated symbols $x_1(b)$ and $x_2(b)$ are the transmitted signals at Sender~1 and Sender~2, respectively.}\label{FigTRIFFC}
\end{figure}
%------------------------------------------------------------------------------------------------------

%-----------------------------------------fig-----------------------------------------------------------

\begin{figure}
\centering
  \includegraphics[width=12.0cm]{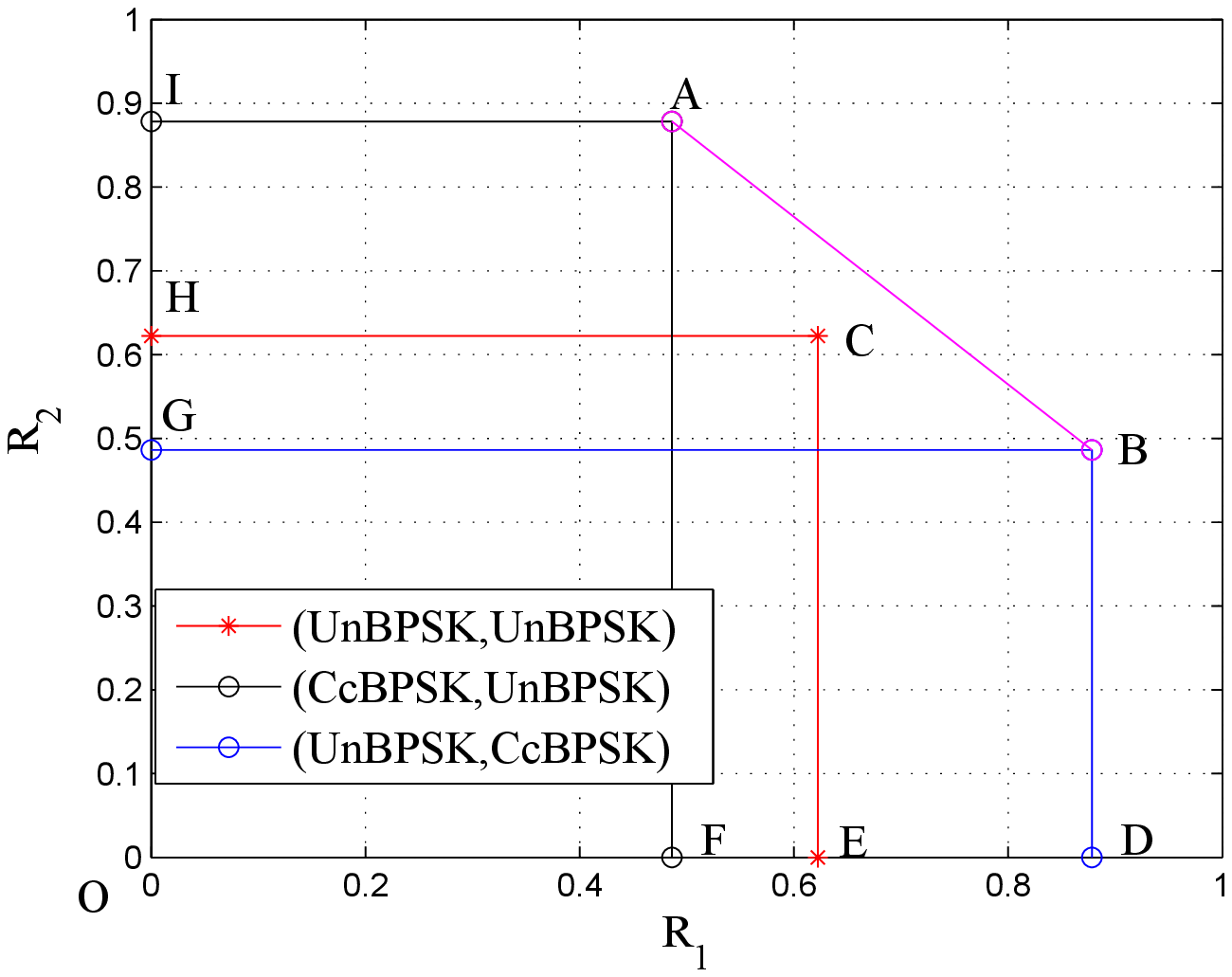}\\
  \caption{The evaluated achievable rate pairs of a specific GIFC, where $P_1=P_2=7.0~{\rm dB}$ and $a = 0.5$. The rectangle OECH with legend ``(UnBPSK, UnBPSK)" corresponds to the case when both senders use UnBPSK as inputs; the rectangle ODBG with legend ``(UnBPSK, CcBPSK)" corresponds to the case when Sender~1 uses UnBPSK as input and Sender~2 uses CcBPSK as input; and the rectangle OFAI with legend ``(CcBPSK, UnBPSK)" corresponds to the case when Sender~1 uses CcBPSK as input and Sender~2 uses UnBPSK as input. }\label{UncodedVSConv_figure}
\end{figure}
%-----------------------------------------fig-----------------------------------------------------------

\section{Decoding Algorithms for Channels with Structured Interference}\label{Sec:decAnddetAlg}

The purpose of this section has two-folds. The first is to present a coding scheme to approach the point ``B" in Fig.~\ref{UncodedVSConv_figure}. The second is to show the decoding gain achieved by taking  into account the structure of the interference.
%The third is to show the advantages of concatenated codes with the convolutional inner codes. \textcolor[rgb]{0.98,0.00,0.00}{Read Gallager's Digital Communication's Book~(Preface), see what is the right writings.}
\subsection{A Coding Scheme}

%In order to verify that knowing the structure of interference will get gains, we design the following coding Schemes.
%Here, we need a kind of code with flexible code-rate. So we use Kite code, proposed in~\cite{Ma11}, which can also be seen as a kind of rateless LDPC codes.

We design a coding scheme using Kite codes\footnote{The main reason that we choose Kite codes is that it is convenient to set up the code rates.
Actually, given data length, the code rates of Kite codes can be ``continuously" varying from $0.1$ to $0.9$ with satisfactory performance, as shown
in~\cite{MaITW11}~\cite{ZhangISIT12}.}. Kite codes are a class of low-density parity-check~(LDPC)
codes, which can be decoded using the sum-product algorithm~(SPA)~\cite{Wiberg95,Kschischang01}. As shown in Fig.~\ref{fig:CodingScheme}, Sender~1 uses a Kite code~(with a parity-check matrix ${\bf H}_1$) and Sender~2 uses a Kite code~(with a parity-check matrix ${\bf H}_2$) concatenated with the convolutional code with the generator matrix
\begin{equation}\nonumber
G(D) = [1 + D + D^2\,\,\,1+D^2].
\end{equation}
%Here, Kite code is a binary linear code with flexible code rate.

%The detailed coding scheme is shown in Fig.~\ref{fig:CodingScheme}.
\begin{figure}[!t]
  \centering
  \includegraphics[width=12.0cm]{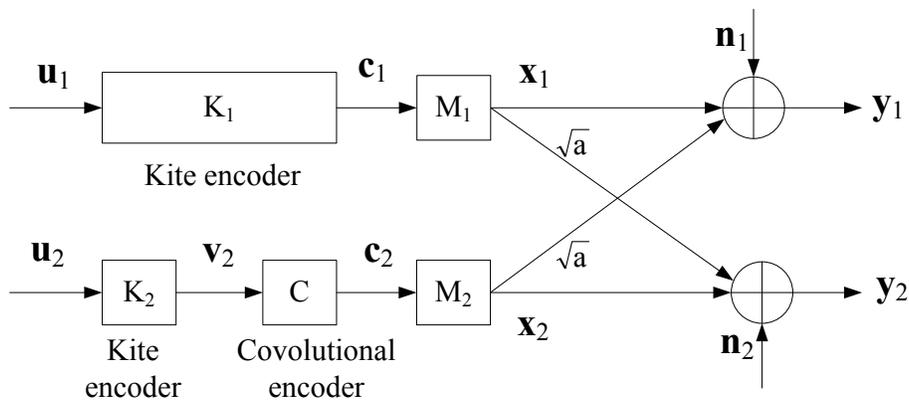}\\
  \caption{A coding scheme for the two-user GIFC.}\label{fig:CodingScheme}
\end{figure}

\emph{Encoding}: For Sender~1, a binary sequence ${\bf u}_1 = (u_{1,1}, u_{1,2},...,u_{1,L_1})$
of length $L_1$ is encoded by a Kite code into a coded sequence ${\bf c}_1 = (c_{1,1}, c_{1,2},...,c_{1,N})$ of length $N$.
For Sender~2, a binary sequence ${\bf u}_2 = (u_{2,1}, u_{2,2},...,u_{2,L_2})$ of length $L_2$
is firstly encoded by a Kite code into a sequence ${\bf v}_2 = (v_{2,1}, v_{2,2},...,v_{2,N'})$ of length $N'$
and then the sequence ${\bf v}_2$ is encoded by the convolutional code with the generator matrix $G(D)$ into
a coded sequence ${\bf c}_2 = (c_{2,1}, c_{2,2},...,c_{2,N})$ of length $N$.

\emph{Modulation}: The codewords ${\bf c}_k$ are mapped into the bipolar sequences
${\bf x}_k = (x_{k,1}, x_{k,2},...,x_{k,N})$ with $x_{k,i} = \sqrt{P_k}(1 - 2c_{k,i})$
($k = 1,2$), where $P_k$ is the power. Then we transmit ${\bf x}_k$ for $k = 1,2$ over the interference channel.

\emph{Decoding}: After receiving ${\bf y}_1$, Receiver~1 attempts to recover the transmitted message ${\bf u}_1$. Similarly,
after receiving ${\bf y}_2$, Receiver~2 attempts to recover the transmitted message ${\bf u}_2$. We will consider several decoding algorithms
in the next subsection to recover the transmitted messages.
%Hence, the next subsection will be used to introduce these algorithms.

\subsection{Decoding Algorithms}
In this subsection, depending on the knowledge about the interference, we present four decoding schemes, including ``knowing only the power of the interference", ``knowing the signaling of the
interference", ``knowing the CC" and ``knowing the whole structure". We focus on the decoding of Receiver~1, while the
decoding of Receiver~2 can be implemented similarly.\footnote{There is no decoding scheme ``Knowing
the CC" for User 2 because User 1 has no convolutional structure.} All these decoding algorithms will be described as\emph{ message processing/passing} algorithms over normal graphs~\cite{Forney01}.

%
%In the following sections, there are four decoding schemes to be shown, including
%. We only
%focus on the decoding of User~1, while the decoding of User~2 is similar to that of User~1
%except that there is no decoding scheme ``Knowing the CC" for User~2, because
%User~1 has no convolutional structure.

%In the following, we introduce general notation for the need of describing the algorithms.
\subsubsection{Message processing/passing algorithms over normal graphs}

\begin{figure}
  \centering
  \includegraphics[width=10.0cm]{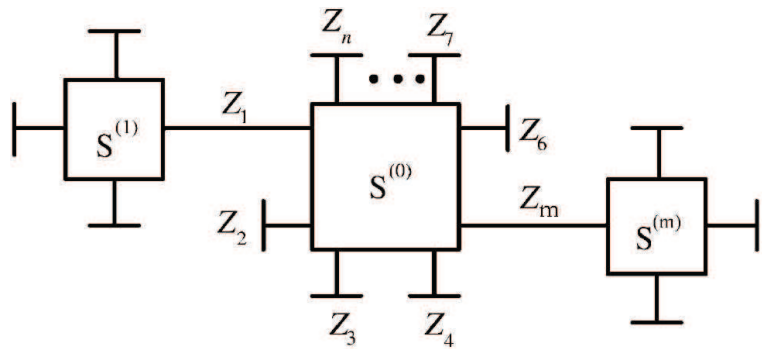}\\
  \caption{A normal graph of a general (sub)system.}
  \label{fig:GeneralNormalGraph}
\end{figure}
%According to Forney~\cite{Forney01}, a normal graph shown in the
%Fig.~\ref{fig:GeneralNormalGraph} can be utilized to represent a decoding
%scheme.
As shown in
Fig.~\ref{fig:GeneralNormalGraph}, a normal graph consists of edges and vertices, which
represent variables and subsystem constraints, respectively. Let
${\bf Z}=\left \lbrace Z_{1},Z_{2},\cdots,Z_{n} \right \rbrace$ be $n$ distinct
random variables that constitute a subsystem $S^{(0)}$. This subsystem
can be represented by a normal subgraph with edges representing ${\bf Z}$ and
a vertex $S^{(0)}$ representing the subsystem constraints. Each {\em half-edge} (ending with a dongle)
may potentially be coupled to some half-edge in other subsystems. For example,
$Z_1$ and $Z_m$ are shown to be connected to subsystems $S^{(1)}$ and $S^{(m)}$,
respectively. In this case, the corresponding edge is called a {\em full-edge}. Associated with each edge is a \emph{message} that is  defined in this paper as the pmf/pdf of the corresponding variable. As in~\cite{Ma12}, we use the notation $P_{Z_i}^{(S^{(i)} \rightarrow S^{(0)})}(z)$ to denote the message from $S^{(i)}$ to $S^{(0)}$.
In particular, we use the notation $P_{Z_i}^{(| \rightarrow S^{(0)})}(z)$ to represent the initial messages ``driving" the subsystem $S^{(0)}$.
For example, such initial messages can be the \emph{a priori} probabilities from the source or the \emph{a posteriori} probabilities computed from the channel observations.
Assume that all messages to $S^{(0)}$ are available.
%Each variable on the edge has a {\em message} associated with it.
%Here, we define the pmf/pdf of a random variable $Z$ as the {\em message}
%transmitted on the edges.
%%%%%%%%%%%%%%%%%%%%%%%%%%%%%%%%%%%%%%%%%%%%%%%%%%%%%%%%%%%%%%%%%%%%%%%%%%%%%%
%For a random variable $Z$ with realizations
%$z \in \mathcal{Z}$, we use $Z \sim P_Z(z)$ to represent that $Z$ has a
%probability mass function (pmf) $\{ P_Z(z), z \in \mathcal{Z} \}$ if $Z$ is
%discrete, or to represent that $Z$ has the probability density function (pdf)
%$\{ P_Z(z), z \in \mathcal{Z} \subseteq \mathbb{R} \}$ if $Z$ is continuous.
%%%%%%%%%%%%%%%%%%%%%%%%%%%%%%%%%%%%%%%%%%%%%%%%%%%%%%%%%%%%%%%%%%%%%%%%%%%%%%
%The pmf of a discrete random variable $Z$ can be regarded as a vector with
%length $|\mathcal{Z}|$, where $|\mathcal{Z}|$ states the number of elements
%in $\mathcal{Z}$.
%As an example, we will describe how to compute the {\em messages} in the normal
%graph.
%Consider the normal graph with vertices (subsystems)
%$\left \lbrace S^{(i)}, i= 0,1,\cdots,m \right \rbrace$ shown in
%Fig.~\ref{fig:GeneralNormalGraph}.
%Let
%${\bf Z}=\left \lbrace Z_{i_1},Z_{i_2},\cdots,Z_{i_n} \right \rbrace$ be $n$
%distinct random variables that connect to vertex $S^{(i)}$. We use the notation
%$P_{Z_{i_j}}^{(S^{(i)} \rightarrow S^{(j)})}(z)$~\cite{Ma12},
%where $z \in \mathcal{Z}$ and $Z_{i_j}$ is the $j$-th edge connecting to vertex
%$S^{(i)}$,~to denote the \emph{message} associated with $Z_{i_j}$
%passing from vertex $S^{(i)}$ to vertex $S^{(j)}$ for any $i$ and $j$. We assume that all incoming messages are available.
The
vertex $S^{(0)}$, as a {\em message processor}, delivers the outgoing message
with respect to any given $Z_{i}$ by computing the likelihood function
\begin{equation}
P_{Z_{i}}^{(S^{(0)} \rightarrow S^{(i)})}(z) \propto
\mbox{Pr} \left \lbrace S^{(0)}
\mbox{ is satisfied } | Z_{i} = z \right \rbrace,
z \in \mathcal{Z}
\label{eq:message_processing}
\end{equation}
by considering all the available messages as well as the system constraints.
We claim that
$P_{Z_{i}}^{(S^{(0)} \rightarrow S^{(i)})}(z)$ is exactly the so-called
{\em extrinsic message} because the computation of the likelihood function is irrelevant to the incoming
message $P_{Z_{i}}^{(S^{(i)} \rightarrow S^{(0)})}(z)$.

\subsubsection{Knowing only the power of the interference}
%%%%%%%%%%%%%%%%%%%%%%%%%%%%%%%%%%%%%%%%%%%%%%%%%%%%%%%%%%%%%%%%%%%%%%%%%%%%%%%%%%
%\begin{figure}[!t]
%  \centering
%  \subfigure[The normal graph for Decoder~1]
%  {
%    \begin{minipage}[t]{0.45\textwidth}
%    \centering
%	\includegraphics[width=3.0in]{NormalGraph/Gaussian_bpsk_user_1.eps}
%	\label{fig:NormalGraphOfGaussianBPSK:user1}
%	\end{minipage}
%  }
%  \vspace{0.0in}
%  \subfigure[The normal graph for decoder~2]
%  {
%    \begin{minipage}[t]{0.45\textwidth}
%    \centering
%	\includegraphics[width=3.0in]{NormalGraph/Gaussian_bpsk_user_2.eps}
%	\label{fig:NormalGraphOfGaussianBPSK:user2}
%	\end{minipage}
%  }
%  \caption{The Normal Graphs for knowing only the power of the interference and the signaling of the interference}
%  \label{fig:NormalGraphOfGaussianBPSK}
%\end{figure}
%%%%%%%%%%%%%%%%%%%%%%%%%%%%%%%%%%%%%%%%%%%%%%%%%%%%%%%%%%%%%%%%%%%%%%%%%%%%%%%%%
\begin{figure}[!t]
  \centering
  \subfigure[]
  {
    \begin{minipage}[t]{0.45\textwidth}
    \centering
	\includegraphics[width=2.3in]{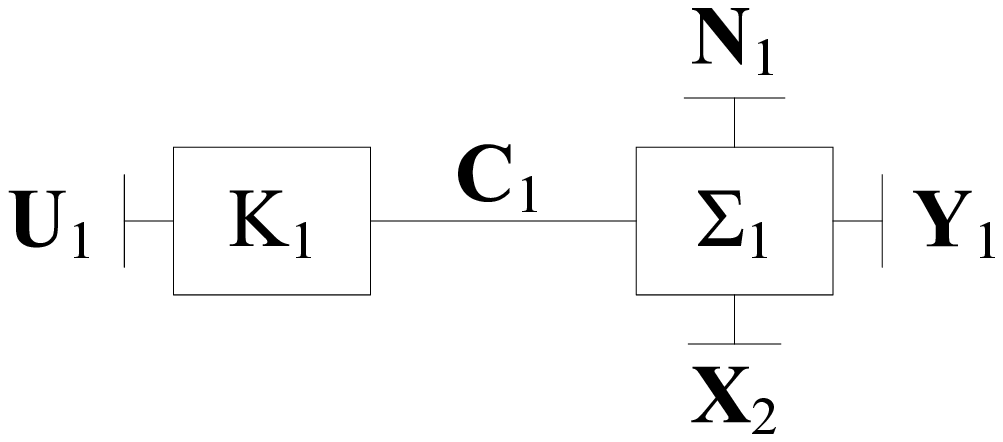}
	\label{fig:NormalGraphOfGB:user1}
	\end{minipage}
  }
\vspace{0.1in}
  \subfigure[]
  {
    \begin{minipage}[t]{0.45\textwidth}
    \centering
	\includegraphics[width=2.0in]{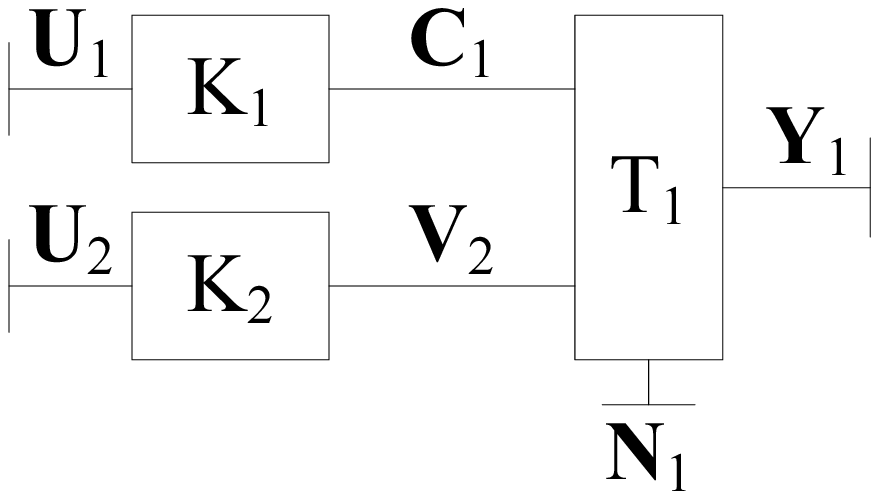}
	\label{fig:NormalGraphOfMAC:user1}
	\end{minipage}
  }

  %\subfigure[]
%  {
%    \begin{minipage}[t]{0.45\textwidth}
%    \centering
%	\includegraphics[width=2.0in]{NormalGraph/mac_user_2.eps}
%	\label{fig:NormalGraphOfMAC:user2}
%	\end{minipage}
%  }
  \caption{The normal graphs: (a) stands for the normal graph of ``knowing only the power of the interference" and ``knowing the signaling of the interference" for Decoder~1; (b) stands for the normal graph of ``knowing the CC" and ``knowing the whole structure" for Decoder~1.}
  \label{fig:NormalGraphOfMAC}
\end{figure}
The decoding scheme for ``knowing only the power of the interference" is the simplest one, which can be described as a message processing/passing algorithm over the normal graph as shown in Fig.~\ref{fig:NormalGraphOfGB:user1}. In this scheme, the
interference from Sender~2 is treated as a Gaussian distribution with mean
zero and variance $aP_2$, where ``$P_2$" is the power and ``$a$" is the square of interference coefficient. That is, Receiver 1 assumes that
$X_{2,j} \sim \mathcal{N}(0, a P_2)$ for $j=1, 2, \cdots, N$.
%We assume that
%\begin{equation}
%U_{1,i} \sim
%\left \lbrace
%P_{U_{1,i}}^{(| \rightarrow K_1)}(0), P_{U_{1,i}}^{(| \rightarrow K_1)}(1)
%\right \rbrace,
%i = 1, 2, \cdots, L_1
%\label{eq:a_priori_message}
%\end{equation}
Since $N_{1,j} \sim \mathcal{N}(0,1)$ for $j=1, 2, \cdots, N$, the decoding algorithm is initialized by the initial messages as follows
%================================================================================
\begin{equation}
\begin{split}
P_{C_{1,j}}^{(\Sigma_1 \rightarrow K_1)}(c) = \mbox{Pr} \left \lbrace C_{1,j} = c | \mathbf{y_1},
X_{2,j} \sim \mathcal{N}(0, a P_2), j=1,2,\cdots,N \right \rbrace \\
\propto
\frac{1}{\sqrt{2 \pi (1+a P_2)}} \exp \left \{ - \frac{[y_{1,j}-\sqrt{P_1}(1-2c)]^2}{2(1+a P_2)} \right\}
, c \in \mathbb{F}_2
\end{split}
\label{eq:InitMsg_C}
\end{equation}
%\begin{equation}
%\begin{split}
%P_{C_{2,j}}^{(e)}(m) = P_{C_{2,j}} \left( m \parallel \Sigma_2, N_{2,j} \sim \mathcal{N}(0,1),
%X_{1,j} \sim \mathcal{N}(0,a P_1), Y_{2,j} \sim \delta(y_2) \right) \\
%\propto
%\frac{1}{\sqrt{2 \pi (1+a P_1)}} \exp \left( - \frac{(y_{2,j}-\sqrt{P_2}(1-2m))^2}{2(1+a P_1)} \right)
%\end{split}
%\label{eq:ExMsg_C}
%\end{equation}
%=================================================================================
for $j=1, 2, \cdots, N$. Then the decoding algorithm uses SPA to compute iteratively the extrinsic messages
$P_{U_{1,i}}^{(K_1 \rightarrow |)}$ and $P_{C_{1,j}}^{(K_1 \rightarrow \Sigma_1)}$. Once these are done, we make the following decisions:
%=================================================================================
\begin{equation}
\hat{u}_{1,i} =
\left \{
\begin{array}{ll}
0,~\mbox{if}~P_{U_{1,i}}^{(| \rightarrow K_1)}(0)P_{U_{1,i}}^{(K_1 \rightarrow |)}(0) >
P_{U_{1,i}}^{(| \rightarrow K_1)}(1)P_{U_{1,i}}^{(K_1 \rightarrow |)}(1), \\
1,~\mbox{otherwise}.
\label{eq:HardDecision_U}
\end{array}
\right.
\end{equation}

\begin{equation}
\hat{c}_{1,j} =
\left \{
\begin{array}{ll}
0,~\mbox{if}~P_{C_{1,j}}^{(\Sigma_1 \rightarrow K_1)}(0)P_{C_{1,j}}^{(K_1 \rightarrow \Sigma_1)}(0) >
P_{C_{1,j}}^{(\Sigma_1 \rightarrow K_1)}(1)P_{C_{1,j}}^{(K_1 \rightarrow \Sigma_1)}(1), \\
1,~\mbox{otherwise}.
\label{eq:HardDecision_C}
\end{array}
\right.
\end{equation}
%\begin{equation}
%\hat{v}_{2,i} =
%\left \{
%\begin{array}{ll}
%0,~\mbox{if}~P_{V_{2,i}}^{(a)}(0)P_{V_{2,i}}^{(e)}(0) >
%P_{V_{2,i}}^{(a)}(1)P_{V_{2,i}}^{(e)}(1), \\
%1,~\mbox{otherwise}.
%\label{eq:HardDecision_V}
%\end{array}
%\right.
%\end{equation}
%================================================================================
%where $\hat{u}_{1,i}$ is the $i$-th bit of the decoded information sequence, and
%$\hat{c}_{1,j}$ is the $j$-th bit of the decoded codeword,
for $i=1,2,\cdots,L_1$
and $j=1,2,\cdots,N$. The details about the decoding algorithm are shown as below.
\begin{algorithm}[``knowing only the power of the
interference"]\label{alg:gaussian}
\mbox{}\par
\begin{itemize}
    \item {\em Initialization:}
    	\begin{enumerate}
			\item {} Initialize $P_{U_{1,i}}^{(| \rightarrow K_1)}(u) = \frac{1}{2}$
			for $i=1,2,\cdots,L_1$ and $u \in \mathbb{F}_2$.
		
			\item {} Compute $P_{C_{1,j}}^{(\Sigma_1 \rightarrow K_1)}(c)$ for $j=1,2,\cdots,N$ and $c \in \mathbb{F}_2$ according to (\ref{eq:InitMsg_C}).
		
		    \item {} Set a maximum iteration number $J$ and iteration variable $j=1$.
		\end{enumerate}
%	\item {} Start the SPA iterative decoding with a maximum iteration number $J$.
    \item {\em Repeat while $j \leq J$:}
		\begin{enumerate}
    		\item {} Compute extrinsic messages $P_{U_{1,i}}^{(K_1 \rightarrow |)}$
    		and $P_{C_{1,j}}^{(K_1 \rightarrow \Sigma_1)}$ for
    		$i=1,2,\cdots,L_1$ and $j=1,2,\cdots,N$ using SPA.

	    	\item {} Make decisions according to
		    (\ref{eq:HardDecision_U}) and (\ref{eq:HardDecision_C}). Denote ${\bf \hat{u}}_1 = \left (\hat{u}_{1,1}, \hat{u}_{1,2}, \cdots, \hat{u}_{1,L_1} \right )$ and
${\bf \hat{c}}_1 = \left (\hat{c}_{1,1}, \hat{c}_{1,2}, \cdots, \hat{c}_{1,N} \right )$.

		    \item {} Compute the syndrome ${\bf S}_1 = {\bf \hat{c}}_1 \cdot {\bf H}_1^T$.
	    	If ${\bf S}_1=\mathbf{0}$, output ${\bf \hat{u}}_1$ and ${\bf \hat{c}}_1$ and exit the iteration.

	    	\item {} Set $j=j+1$. If ${\bf S}_1 \neq \mathbf{0}$ and $j>J$,
    		report a decoding failure.
		\end{enumerate}
	\item {\em End decoding}.
\end{itemize}
\end{algorithm}

\subsubsection{Knowing the signaling of the interference}
The decoding algorithm for this scheme is almost the same as
Algorithm~\ref{alg:gaussian}, see Fig.~\ref{fig:NormalGraphOfGB:user1}. The difference is that $X_{2,j} \sim \mathcal{B}(1/2)$~(Bernoulli-1/2 distribution\footnote{Strictly speaking, $X_{2,j}$ is a shift/scaling version of $\mathcal{B}(1/2)$.}) for $j=1,2,\cdots,N$. So the computation of
$P_{C_{1,j}}^{(\Sigma_1 \rightarrow K_1)}(c)$ is changed into
\begin{equation}
\begin{split}
P_{C_{1,j}}^{(\Sigma_1 \rightarrow K_1)}(c) = \mbox{Pr} \left \lbrace C_{1,j} = c | \mathbf{y_1},
X_{2,j} \sim \mathcal{B}(1/2), j=1,2,\cdots,N \right \rbrace \\
\propto
\frac{1}{2}\frac{1}{\sqrt{2 \pi}} \exp \left \lbrace - \frac{ \left[ y_{1,j}-\sqrt{P_1}(1-2c)-\sqrt{a P_2} \right]^2}{2} \right \rbrace\\
+
\frac{1}{2}\frac{1}{\sqrt{2 \pi}} \exp \left \lbrace - \frac{ \left[ y_{1,j}-\sqrt{P_1}(1-2c)+\sqrt{a P_2} \right]^2}{2} \right \rbrace,\\
c \in \mathbb{F}_2
\end{split}
\label{eq:InitMsg_C_BPSK}
\end{equation}
Then the decoding algorithm of ``knowing the signaling of the interference" can be shown as below.
\begin{algorithm}[``knowing the signaling of the
interference"]\label{alg:bpsk}
\mbox{}\par
\begin{itemize}
    \item {\em Initialization:}
    	\begin{enumerate}
			\item {} Initialize $P_{U_{1,i}}^{(| \rightarrow K_1)}(u) = \frac{1}{2}$
			for $i=1,2,\cdots,L_1$ and $u \in \mathbb{F}_2$.
		
			\item {} Compute $P_{C_{1,j}}^{(\Sigma_1 \rightarrow K_1)}(c)$	for $j=1,2,\cdots,N$ and $c \in \mathbb{F}_2$ according to (\ref{eq:InitMsg_C_BPSK}).
		
		    \item {} Set a maximum iteration number $J$ and iteration variable $j=1$.
		\end{enumerate}
%	\item {} Start the SPA iterative decoding with a maximum iteration number $J$.
    \item {\em Repeat while $j \leq J$:}
		\begin{enumerate}
    		\item {} Compute extrinsic messages $P_{U_{1,i}}^{(K_1 \rightarrow |)}$
    		and $P_{C_{1,j}}^{(K_1 \rightarrow \Sigma_1)}$ for
    		$i=1,2,\cdots,L_1$ and $j=1,2,\cdots,N$ using SPA.

	    	\item {}  Make decisions according to
		    (\ref{eq:HardDecision_U}) and (\ref{eq:HardDecision_C}), respectively.

		    \item {} Compute the syndrome ${\bf S}_1 = {\bf \hat{c}}_1 \cdot {\bf H}_1^T$.
	    	If ${\bf S}_1=\mathbf{0}$, output ${\bf \hat{u}}_1$ and ${\bf \hat{c}}_1$ and exit the iteration.

	    	\item {} Set $j=j+1$. If ${\bf S}_1 \neq \mathbf{0}$ and $j>J$,
    		report a decoding failure.
		\end{enumerate}
	\item {\em End decoding}.
\end{itemize}
\end{algorithm}

\subsubsection{Knowing the CC}
%\begin{figure}[!t]
%  \centering
%  \includegraphics[width=3.0in]{NormalGraph/conv_user_1.eps}
%  \caption{The Normal Graph for Knowing the CC}
%  \label{fig:NormalGraphOfConvUser1}
%\end{figure}
``Knowing the CC" means that Decoder~1 knows the structure of the convolutional code. This scheme can be described as a message processing/passing algorithm over the normal graph as shown in
Fig.~\ref{fig:NormalGraphOfMAC:user1}. Actually, the
vertex $T_1$ is a combination of three subsystems, convolutional
encoder, modulation and GIFC constraint, which can be specified by a trellis $\mathcal{T}$ with parallel branches~\cite{Huang2011}. Therefore, the BCJR
algorithm can be used to compute the extrinsic messages
$P_{C_{1,j}}^{(T_1 \rightarrow K_1)}(c)$ for $j=1,2,\cdots,N$ over the trellis $\mathcal{T}$.
%The important part of implementing the BCJR algorithm is to compute the
%metric of each branch in the trellis. We use a convolutional code with a
%generator matrix described in (\ref{Generator}) in our simulations,
%According to~(\ref{eq:branch_metric}), the metric for a branch $b$ can be
%assigned as
%\begin{equation}
%\rho \left( b \right) = \frac{1}{4}
%\end{equation}
Since the structure of Kite code for Sender~2 is unknown, the constraint of
vertex $K_2$ is inactive. In this case, the pmf of variable $V_{2,k}$ ($k=1,2,\cdots,N'$)
is assumed to be Bernoulli-1/2 distribution. There are two strategies to
implement the BCJR algorithm. One is called ``BCJR-once", in which the BCJR algorithm is performed only once. The other strategy is called ``BCJR-repeat", in which the BCJR
algorithm is performed more than once. In this scheme, the decoding decisions on $C_{1,j}$ are modified into
\begin{equation}
\hat{c}_{1,j} =
\left \{
\begin{array}{ll}
0,~\mbox{if}~P_{C_{1,j}}^{(T_1 \rightarrow K_1)}(0)P_{C_{1,j}}^{(K_1 \rightarrow T_1)}(0) >
P_{C_{1,j}}^{(T_1 \rightarrow K_1)}(1)P_{C_{1,j}}^{(K_1 \rightarrow T_1)}(1), \\
1,~\mbox{otherwise},
\label{eq:HardDecision_C_CC}
\end{array}
\right.
\end{equation}
for $j=1,2,\cdots,N$. These two decoding procedures are described in
Algorithm \ref{alg:bcjr1} and Algorithm \ref{alg:conv}, respectively.

\begin{algorithm}[BCJR-once]\label{alg:bcjr1}
\mbox{}\par
\begin{itemize}
    \item {\em Initialization:}
    	\begin{enumerate}
	    	\item {} Initialize pmf
		    $P_{C_{1,j}}^{(K_1 \rightarrow T_1)} \left( c \right) = \frac{1}{2}$ and
		    $P_{C_{2,j}}^{(| \rightarrow T_1)} \left( c \right) = \frac{1}{2}$
		    for $j=1,2,\cdots,N, c \in \mathbb{F}_2$ and
		    $P_{V_{2,k}}^{(| \rightarrow T_1)} \left( v \right) = \frac{1}{2}$
		    for $k=1,2,\cdots,N', v \in \mathbb{F}_2$.
	
		    \item {} Compute extrinsic messages
	    	$P_{C_{1,j}}^{(T_1 \rightarrow K_1)} \left( c \right)$ for $j=1,2,\cdots,N$, $c \in \mathbb{F}_2$ using BCJR algorithm over the parallel branch trellis $\mathcal{T}$.
	
	    	\item {} Set a maximum iteration number $J$ and iteration variable $j=1$.
		\end{enumerate}
	
    \item {\em Repeat while $j \leq J$:}
		\begin{enumerate}
    		\item {} Compute extrinsic messages $P_{U_{1,i}}^{(K_1 \rightarrow |)}$
    		and $P_{C_{1,j}}^{(K_1 \rightarrow T_1)}$ for
    		$i=1,2,\cdots,L_1$ and $j=1,2,\cdots,N$ using SPA.

	    	\item {} Make decisions according to
		    (\ref{eq:HardDecision_U}) and (\ref{eq:HardDecision_C_CC}).

		    \item {} Compute the syndrome ${\bf S}_1 = {\bf \hat{c}}_1 \cdot {\bf H}_1^T$.
	    	If ${\bf S}_1=\mathbf{0}$, output ${\bf \hat{u}}_1$ and ${\bf \hat{c}}_1$ and exit the iteration.

	    	\item {} Set $j=j+1$. If ${\bf S}_1 \neq \mathbf{0}$ and $j>J$,
    		report a decoding failure.
		\end{enumerate}
	
	\item {\em End Decoding}
\end{itemize}
\end{algorithm}

\begin{algorithm}[BCJR-repeat]\label{alg:conv}
\mbox{}\par
\begin{itemize}
    \item {\em Initialization:}
    	\begin{enumerate}
	    	\item {} Initialize pmf
		    $P_{C_{1,j}}^{(K_1 \rightarrow T_1)} \left( c \right) = \frac{1}{2}$ and
		    $P_{C_{2,j}}^{(| \rightarrow T_1)} \left( c \right) = \frac{1}{2}$
		    for $j=1,2,\cdots,N, c \in \mathbb{F}_2$ and
		    $P_{V_{2,k}}^{(| \rightarrow T_1)} \left( v \right) = \frac{1}{2}$
		    for $k=1,2,\cdots,N', v \in \mathbb{F}_2$.
	
	    	\item {} Set a maximum iteration number $J$ and iteration variable $j=1$.
		\end{enumerate}

    \item {\em Repeat while $j \leq J$:}
		\begin{enumerate}		
		    \item {} Compute extrinsic messages
	    	$P_{C_{1,j}}^{(T_1 \rightarrow K_1)} \left( c \right)$ for $j=1,2,\cdots,N$, $c \in \mathbb{F}_2$ using BCJR algorithm over the parallel branch
		    trellis $\mathcal{T}$.
		
    		\item {} Compute extrinsic messages $P_{U_{1,i}}^{(K_1 \rightarrow |)}$
    		and $P_{C_{1,j}}^{(K_1 \rightarrow T_1)}$ for
    		$i=1,2,\cdots,L_1$ and $j=1,2,\cdots,N$ using SPA.

	    	\item {} Make decisions according to
		    (\ref{eq:HardDecision_U}) and (\ref{eq:HardDecision_C_CC}).

		    \item {} Compute the syndrome ${\bf S}_1 = {\bf \hat{c}}_1 \cdot {\bf H}_1^T$.
	    	If ${\bf S}_1=\mathbf{0}$, output ${\bf \hat{u}}_1$ and ${\bf \hat{c}}_1$ and exit the iteration.

	    	\item {} Set $j=j+1$. If ${\bf S}_1 \neq \mathbf{0}$ and $j>J$,
    		report a decoding failure.
		\end{enumerate}
	
	\item {\em End Decoding}
\end{itemize}
\end{algorithm}

\subsubsection{Knowing the whole structure}

The scheme ``knowing the whole structure" for Receiver~1 can also be described as a message processing/passing algorithm over the normal graph shown in
Fig.~\ref{fig:NormalGraphOfMAC:user1}. Since knowing the whole structure of
the interference, Receiver~1 can decode iteratively utilizing the structure of both
users. Using the
BCJR algorithm, $P_{C_{1,j}}^{(T_1 \rightarrow K_1)} \left( c \right)$ and
$P_{V_{2,k}}^{(T_1 \rightarrow K_2)} \left( v \right)$ are computed simultaneously
over the parallel branch trellis $\mathcal{T}$. The iterative decoding algorithm is presented
in Algorithm~\ref{alg:mac}.

\begin{algorithm}[``knowing the whole structure"]\label{alg:mac}
\mbox{}\par
\begin{itemize}
    \item {\em Initialization:}
    	\begin{enumerate}
	    	\item {} Initialize pmf
		    $P_{C_{1,j}}^{(K_1 \rightarrow T_1)} \left( c \right) = \frac{1}{2}$ and
		    $P_{C_{2,j}}^{(| \rightarrow T_1)} \left( c \right) = \frac{1}{2}$
		    for $j=1,2,\cdots,N, c \in \mathbb{F}_2$ and
		    $P_{V_{2,k}}^{(K_2 \rightarrow T_1)} \left( v \right) = \frac{1}{2}$
		    for $k=1,2,\cdots,N', v \in \mathbb{F}_2$.
	
	    	\item {} Set a maximum iteration number $J$ and iteration variable $j=1$.
		\end{enumerate}

    \item {\em Repeat while $j \leq J$:}
		\begin{enumerate}		
		    \item {} Compute extrinsic messages
	    	$P_{C_{1,j}}^{(T_1 \rightarrow K_1)} \left( c \right)$ for $j=1,2,\cdots,N$,
	    	$c \in \mathbb{F}_2$ and
		    $P_{V_{2,k}}^{(T_1 \rightarrow K_2)} \left( v \right)$ for $k=1,2,\cdots,N'$,
		    $v \in \mathbb{F}_2$ using BCJR algorithm over the parallel branch trellis $\mathcal{T}$.
		
    		\item {} Compute extrinsic messages $P_{U_{1,i}}^{(K_1 \rightarrow |)}$
    		and $P_{C_{1,j}}^{(K_1 \rightarrow T_1)}$ for
    		$i=1,2,\cdots,L_1$ and $j=1,2,\cdots,N$ using SPA.
    		
    		\item {} Compute extrinsic messages
		    $P_{V_{2,k}}^{(K_2 \rightarrow T_1)} \left( v \right)$
		    for $k=1,2,\cdots,N', v \in \mathbb{F}_2$ using SPA.

	    	\item {} Make decisions according to
		    (\ref{eq:HardDecision_U}) and (\ref{eq:HardDecision_C_CC}).

		    \item {} Compute the syndrome ${\bf S}_1 = {\bf \hat{c}}_1 \cdot {\bf H}_1^T$.
	    	If ${\bf S}_1=\mathbf{0}$, output ${\bf \hat{u}}_1$ and ${\bf \hat{c}}_1$ and exit the iteration.

	    	\item {} Set $j=j+1$. If ${\bf S}_1 \neq \mathbf{0}$ and $j>J$,
    		report a decoding failure.
		\end{enumerate}
	
	\item {\em End Decoding}
\end{itemize}
\end{algorithm}
%In the above subsections, the four decoding algorithms for User~1 have been introduced. For User~2, the decoding
%procedures are similar except that the scheme ``knowing the CC" is not available for User~2.

\subsection{Numerical Results}
%In all simulations, we Set the maximum iteration $J=200$.
In this subsection, simulation results of the decoding algorithms are shown and analyzed.
%we has to use the exact code rate pair of the boundary of Fig. \ref{UncodedVSConv_figure}.
%And we design a kind of high performance LDPC code called Kite Code, which can
%construct arbitrary given rational code rate easily, to approach the theoretical performance.
Simulation parameters of Fig.~\ref{fig:User1performance} and Fig.~\ref{fig:User2performance} are presented in TABLE \ref{tab:simulation_parameter}. In these two
figures, we
let the power constraints of two senders be same, that is, $P_1 \equiv P_2 = P$. Here, ``Gaussian" stands for the scheme ``knowing only the power of the
interference", ``BPSK" stands for the scheme ``knowing the signaling of the interference", ``BCJR1" stands for the scheme ``BCJR-once", ``CONV" stands for the scheme ``BCJR-repeat"
and ``Know All Structure" stands for the scheme ``knowing the whole structure". From Fig.~\ref{fig:User1performance} and
Fig.~\ref{fig:User2performance}, we can easily see that the decoding gains get larger as
more details of the structure of the interference are known. %But in
%Fig.~\ref{fig:User1performance}, we find out the performances of knowing the CC and
%knowing the whole structure are approximate.

\begin{center}
\begin{threeparttable}[!t]
\caption{Parameters of the BER performance simulations}
\label{tab:simulation_parameter}
\centering
% Some packages, such as MDW tools, offer better commands for making tables
% than the plain LaTeX2e tabular which is used here.
\begin{tabular}{c|ccccc}
\hline
\hline
Parameters                  & Values \\
\hline
Square of interference coefficient $a$ & 0.5 \\
Maximum iteration number $J$  & 200 \\
Kite Code of Sender~1   & $N=10000,L_1=8782$ \\
Kite Code of Sender~2   & $N'=5000,L_2=4862$ \\
Generator matrix $G(D)$ & $[1 + D + D^2\,\,\,1+D^2]$ \\
Code rate pair $\left( R_1,R_2 \right)$ & $\left( 0.8782,0.4862 \right)$ \\
\hline
\hline
\end{tabular}
\end{threeparttable}
\end{center}

\begin{figure}[!t]
  \centering
  \includegraphics[width=12.0cm]{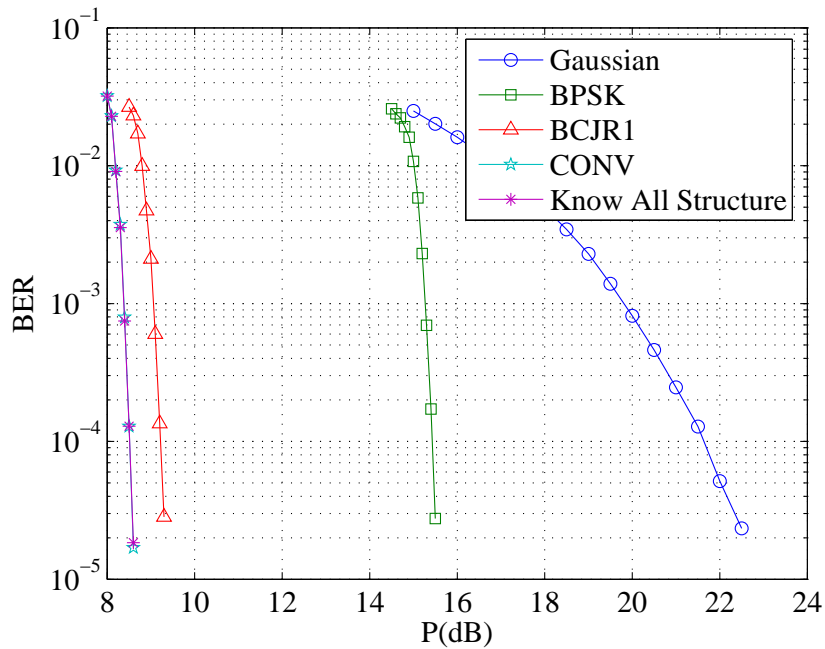}\\
  \caption{The error performance of Receiver~1.  ``Gaussian" stands for the scheme ``knowing only the power of the
interference", ``BPSK" stands for the scheme ``knowing the signaling of the interference", ``BCJR1" stands for the scheme ``BCJR-once", ``CONV" stands for the scheme ``BCJR-repeat"
and ``Know All Structure" stands for the scheme ``knowing the whole structure".}
  \label{fig:User1performance}
\end{figure}

\begin{figure}[!t]
  \centering
  \includegraphics[width=12.0cm]{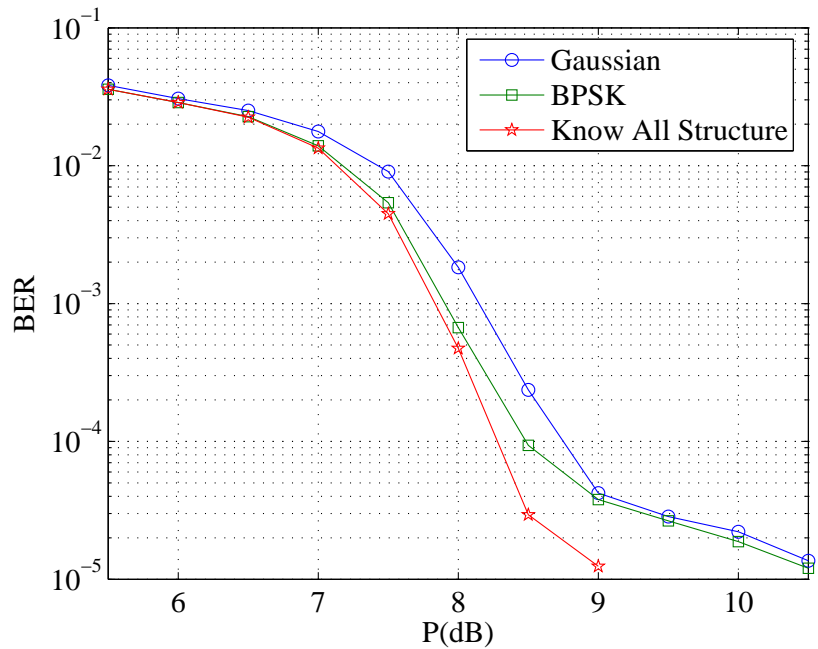}\\
  \caption{The error performance of Receiver~2. ``Gaussian" stands for the scheme ``knowing only the power of the
interference", ``BPSK" stands for the scheme ``knowing the signaling of the interference"
and ``Know All Structure" stands for the scheme ``knowing the whole structure".}
  \label{fig:User2performance}
\end{figure}

Another objective is to find out a code rate pair nearest to the point ``B" in Fig.~\ref{UncodedVSConv_figure} with
bit error rate~(BER) performance of $10^{-4}$. So we do the simulations with different code rate pairs. In
the simulations, we adopt the scheme ``knowing the whole structure" and gradually decrease the code rates from the point ``B" with a
step length $0.01$. Simulation parameters for different code rate pairs are
listed in TABLE \ref{tab:simulation_paramete_coderate}, while the simulation
results are presented using a 3D graph in Fig.~\ref{fig:3dgraph}. From the
figure, it is obvious that as the code rates of two users are decreasing, the
BER also decreases. Finally, we find out the ``best" code rate pair is
$(0.71, 0.48)$ for User~1 and User~2. The theoretical value of the point ``B" is about $(0.878,0.486)$. So we can see that the gap between the result using our decoding scheme and the theoretical value is small.

\begin{center}
\begin{threeparttable}[!t]
\caption{Parameters of the simulations for different code rate pairs.}
\label{tab:simulation_paramete_coderate}
\centering
% Some packages, such as MDW tools, offer better commands for making tables
% than the plain LaTeX2e tabular which is used here.
\begin{tabular}{c|ccccc}
\hline
\hline
Parameters                    & Values \\
\hline
Square of interference coefficient $a$ & $0.5$ \\
Maximum iteration number $J$        & $200$ \\
Code length $N$ of Kite Code of Sender~1  & $10000$ \\
Code length $N'$ of Kite Code of Sender~2  & $5000$ \\
Generator matrix $G(D)$  & $[1 + D + D^2\,\,\,1+D^2]$ \\
Step length            & $100$ \\
Range of message length $L_1$ & $7100 \sim 8800$ \\
Range of message length $L_2$ & $4000 \sim 4900$ \\
\hline
\hline
\end{tabular}
\end{threeparttable}
\end{center}

\begin{figure}[!t]
  \centering
  \includegraphics[width=12.0cm]{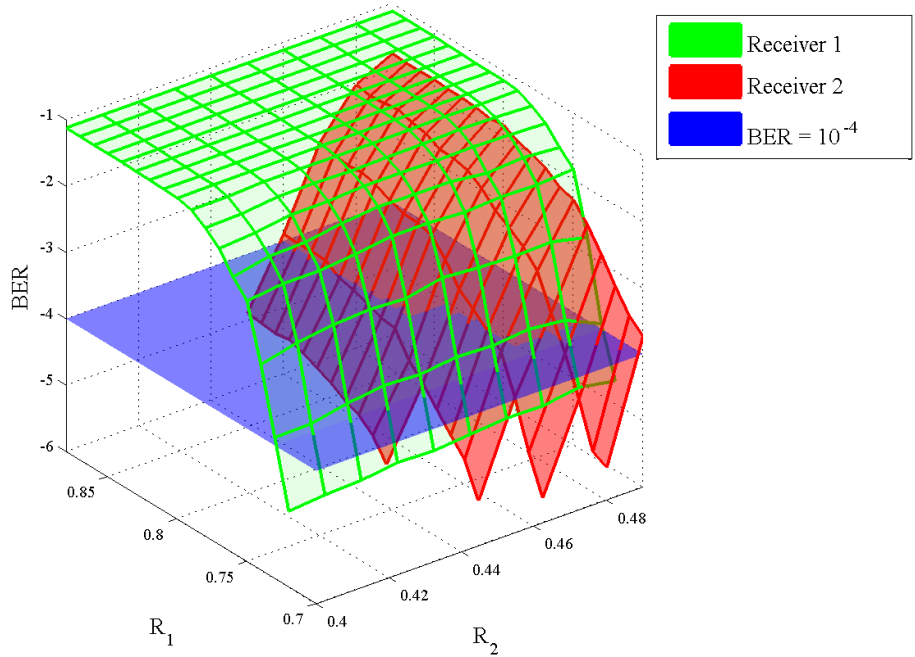}\\
  \caption{Error performance of two users with different code rate pairs $(R_1,R_2)$. Blue plane represents BER level, green surface stands for the error performance of Receiver~1 and red surface stands for the error performance of Receiver~2.}
  \label{fig:3dgraph}
\end{figure}

\section{Conclusions}\label{Conclusions}
In this paper, we have proved that the capacity region of the two-user interference channel is the union of a family of rectangles, each of which is defined by a pair of spectral inf-mutual information rates associated with two independent input processes. For the stationary memoryless channel with discrete Markov inputs, the defined pair of rates can be computed, which show us that the simplest inner bounds~(obtained by treating the interference as noise) could be improved by taking into account the structure of the interference processes. Also a concrete coding scheme to approach the theoretical achievable rate pairs was presented, which showed that the decoding gain can be achieved by considering the structure of the interference.

\vspace{-0.25cm}
\bibliographystyle{IEEEtran}
\bibliography{IEEEabrv,mybibfile}

%\end{thebibliography}

\end{document}